\documentclass[11pt]{article}
\usepackage[a4paper,left=4cm,right=4cm,top=4cm,bottom=5cm]{geometry}

\usepackage[numbers]{natbib}

\usepackage{setspace}

\usepackage{stmaryrd}
\usepackage{amsmath}
\usepackage{amsthm}
\usepackage{amsfonts}
\usepackage{courier}
\usepackage{dsfont}
\usepackage{enumerate}
\usepackage{enumitem}
\usepackage{graphicx}
	\usepackage{epstopdf}
\usepackage{latexsym}
\usepackage{url}
\usepackage{bbm}

\usepackage{mathtools}
\mathtoolsset{showonlyrefs}
 
\usepackage{todonotes} 

\usepackage{caption}
\theoremstyle{plain}
	\newtheorem{prop}{Proposition}[section]
	
	\newtheorem{theo}[prop]{Theorem}
	\newtheorem{lem}[prop]{Lemma}

\theoremstyle{definition} 
	
	\newtheorem{ex}[prop]{Example}

\newcommand{\eps}{\varepsilon}

\newcommand{\diff}{d}
\newcommand{\wt}{\widetilde}

\newcommand{\eqspace}{\hspace{4em}}
\newcommand{\Cov}{\text{cov}}
\newcommand{\Var}{\text{var}}
\newcommand{\R}{\mathbb{R}}
\newcommand{\E}{\mathbb{E}}
\newcommand{\N}{\mathbb{N}}
\newcommand{\Z}{\mathbb{Z}}

\renewcommand{\P}{\mathbb P}
\newcommand{\HD}{\text{HD}}
\newcommand{\toop}{\stackrel{\P}{\longrightarrow}}
\DeclareMathOperator*{\argmin}{arg\,min}

\newcommand{\dn}{\boxempty_n}
\newcommand{\tr}{\lhd\,} 
\newcommand{\ol}{\overline}
\newcommand{\sq}{\boxempty\,}

\newcommand{\be}{\mathbf e}
\newcommand{\bt}{\mathbf t}
\newcommand{\bx}{\mathbf x}

\newcommand{\bss}{\mathbf s}
\newcommand{\bu}{\mathbf u}
\newcommand{\bj}{\mathbf j}
\newcommand{\bi}{\mathbf i}
\newcommand{\bk}{\mathbf k}

\newcommand{\bb}{\mathbf b}
\newcommand{\br}{\mathbf r}

\begin{document}

\title{Hybrid simulation scheme for volatility modulated moving average fields}
\author{Claudio Heinrich \thanks{Norsk Regnesentral Oslo, 
E-mail: claudio.heinrich@nr.no}  \and
Mikko S. Pakkanen \thanks{Department
of Mathematics, Imperial College London, 
E-mail: m.pakkanen@imperial.ac.uk}  \and
Almut E.D. Veraart\thanks{Department
of Mathematics, Imperial College London,
E-mail: a.veraart@imperial.ac.uk}}

\maketitle

\begin{abstract}
We develop a simulation scheme for a class of spatial stochastic processes called volatility modulated moving averages. A characteristic feature of this model is that the behaviour of the moving average kernel at zero governs the roughness of realisations, whereas its behaviour away from zero determines the global properties of the process, such as long range dependence. Our simulation scheme takes this into account and approximates the moving average kernel by a power function around zero and by a step function elsewhere. For this type of approach the authors of \cite{BenLunPak2016}, who considered an analogous model in one dimension, coined the expression hybrid simulation scheme.
We derive the asymptotic mean square error of the simulation scheme and compare it in a simulation study with several other simulation techniques and exemplify its favourable performance in a simulation study.
\end{abstract}




\section{Introduction}

In this article we develop a simulation scheme for real-valued random fields that we call \emph{volatility modulated moving average} (VMMA) fields. A VMMA is defined by the formula
\begin{align}\label{VMMA}
X_{\bt}=\int_{\R^2} g(\bt-\bss)\sigma_{\bss}W(d\bss),
\end{align}
 where $W$ is Gaussian white noise, $g\in L^2(\R^2)$ is a deterministic kernel, and $\sigma$ is a random volatility field.
This model has been used for statistical modelling of spatial phenomena in various disciplines, examples being modelling of vegetation and nitrate deposition \cite{HuaWanBreDav2011}, of sea surface temperature \cite{NguVer2017} and of wheat yields \cite{Yan2007}.

We are interested in the case where the moving average kernel $g$ has a singularity at zero. In this situation, the order of the singularity governs the roughness of the random field, specified by its Hausdorff dimension or index of H\"older continuity. 
Spatial stochastic models with Hausdorff dimension greater 2 (i.e. with non-smooth realisations) are for example used in surface modelling, where it is of high importance to model the roughness of the surface accurately. Specific examples are modelling of seafloor morphology \cite{GofJor1988} or surface modelling of celestial bodies \cite{HanThoOvcGneRic2015}.

A particular challenge in simulating volatility modulated moving averages lies in recovering the roughness of the field, while simultaneously capturing the global properties of the field, such as long range dependence. 
Our hybrid simulation scheme relies on approximating the kernel $g$ by a power function in a small neighbourhood of zero, and by 
a step function away from zero. This approach allows us to reproduce the explosive behaviour at the origin, while simultaneously approximating the integrand on a large subset of $\R^2.$
This idea is motivated by the recent work \cite{BenLunPak2016}, where the authors proposed an analogous scheme for the simulation of the one-dimensional model of Brownian semi-stationary processes. As a consequence, the hybrid simulation scheme preserves the roughness of the random field.

 It is known that any stationary Gaussian random field with a continuous and integrable covariance function has a moving average representation of the form \eqref{VMMA} with $\sigma$ constant, cf. \cite[Proposition 6]{HelProJen2008}.
This is for example satisfied for stationary Gaussian fields with Mat\'ern covariance, see Remark \ref{Matern} for details.
 In the literature, much attention has been devoted to the case when the roughness or shape parameter $\nu$ of the Mat\'ern model is integer valued, and to the cases $\nu=\frac 3 2$ and $\nu=\frac 5 2$, where the covariance function is often referred to as second- and third-order autoregressive function. When $\nu$ is integer-valued, the Gaussian field can be approximated by a Gaussian Markov random field which can be efficiently simulated, see \cite{LinRueLin2011} for details. Their approach cannot be applied when $\nu\in(0,1),$ which is the case under consideration in this paper. This rough Mat\'ern model has for example been used in \cite{GofJor1988} and in the context of turbulence modelling \cite{vonKar1948}, where the value $\nu=1/3$ is of particular interest.
 Introducing the stochastic volatility factor $\sigma$ allows for modelling spatial heteroscedasticity and non-Gaussian marginal distributions.
In particular, when $\sigma$ is covariance stationary and independent of $W$, and $g$ is as specified in Remark \ref{Matern}, the field $X$ is non-Gaussian with Mat\'ern covariance.
This is an alternative way of constructing non-Gaussian Mat\'ern covariance fields to the the more general approach taken in the recent publications \cite{Bol2014,WalBol2015}.

In a simulation study, we compare the hybrid simulation scheme to other simulation methods for the model \eqref{VMMA}, namely to what we call the Riemann-sum scheme, which corresponds to approximating the integrand by a step function, and to exact simulation using circulant embedding of the covariance matrix, as described in \cite{DieNew1993,WooCha1994}. The hybrid scheme is not exact, as it approximates the integrand only on a compact set. However, using circulant embeddings requires the process to be Gaussian and stationary, which the model \eqref{VMMA} only satisfies in some special cases, for example when $\sigma$ is constant. Moreover, in order to apply exact simulation methods, the covariance function of $X$ needs to be known, which is oftentimes costly to compute from the model \eqref{VMMA}.
In theory, the asymptotic computational costs of the hybrid scheme are slightly higher than for the circulant embeddings method, as $n\to\infty$, i.e. as the simulation grid gets finer, see Section \ref{secHS} for details. However, we found in our simulation study that for a wide range of parameters the hybrid scheme performs in fact faster than exact simulation, even for large values of $n$, see Table \ref{ComTim}.

This article is structured as follows. In Section \ref{secMod} we introduce our model in detail and discuss some of its properties. In Section \ref{secHS} we describe the hybrid simulation scheme and derive the asymptotic error of the scheme.  Section \ref{secNum} contains the simulation study comparing the hybrid scheme to other simulation schemes. Proofs  for our theoretical results are given in Section \ref{secPro}. The appendix contains some technical details and calculations.


\section{Volatility modulated moving average fields}
\label{secMod}


Let $(\Omega,\mathcal F,\P)$ be a probability space, and $W$ white noise on $\R^2$. That is, $W$ is an independently scattered random measure satisfying $W(A)\sim \mathcal N(0,\lambda(A))$ for all sets $A\in \mathcal B_0=\{A \in \mathcal B(\R^2)\,:\, \lambda(A)<\infty\}$, where $\lambda$ denotes the Lebesgue measure. Recall that a collection of real valued random variables $\Lambda=\{\Lambda(A)\,:\, A\in \mathcal B_0\}$ is called independently scattered random measure if for every sequence $(A_n)_{n\in\N}$ of disjoint sets with $\lambda(\bigcup_n A_n)<\infty$, the random variables $\Lambda(A_n), n=1,2,...$ are independent and $\Lambda(\bigcup_n A_n)=\sum_n \Lambda(A_n), $ almost surely.

The kernel function $g:\R^2\to\R$ is assumed to be of the form
\begin{align}\label{gAss}
g(\bt)= \wt g(\|\bt\|):=\|\bt\|^\alpha L(\|\bt\|)
\end{align}
for some $\alpha\in (-1,0)$, and a function $L:(0,\infty)\to(0,\infty)$ that is slowly varying at 0. Here and in the following $\|\cdot\|$ always denotes the Euclidean norm on $\R^2.$
Recall that $L$ is said to be slowly varying at 0 if for any $\delta>0$
\[\lim_{x\to 0} \frac{L(\delta x)}{L(x)}=1,\]
and that then the function $\wt g (x)=x^\alpha L(x)$ is called regularly varying at 0 of index $\alpha$. 
The explosive behaviour of the kernel at 0 is a crucial feature of this model, as it governs the roughness of the field.
Indeed, under weak additional assumptions the Hausdorff dimension of a realisation of $X$ is $2-\alpha$ with probability 1, see \cite{HanTho2013} and Theorem \ref{Roughness}, meaning that for $\alpha\to -1$ the realisations of $X$ become extremely rough.
In Figure \ref{figSam} we present samples of realisations of VMMAs for different $\alpha$.  

The roughness of realisations poses a challenge for simulation of volatility modulated moving averages. Indeed, possibly the most intuitive way to simulate the model \eqref{VMMA} is by freezing the integrand over small blocks and simulating the white noise over these blocks as independent centered normal random variables with variance equaling the block size. However, this method does not account for the explosive behaviour of $g$ at 0 and therefore does a poor job in reproducing the roughness of the original process correctly, in particular for values of $\alpha$ close to $-1$. We will demonstrate this phenomenon in a simulation study in Section \ref{secNum}.
The hybrid scheme resolves this issue by approximating $g$ around 0 by a power kernel, and approximating it by a step function away from 0. 

\begin{figure}
 \includegraphics[width=0.42 \textwidth, height = 0.2 \textheight]{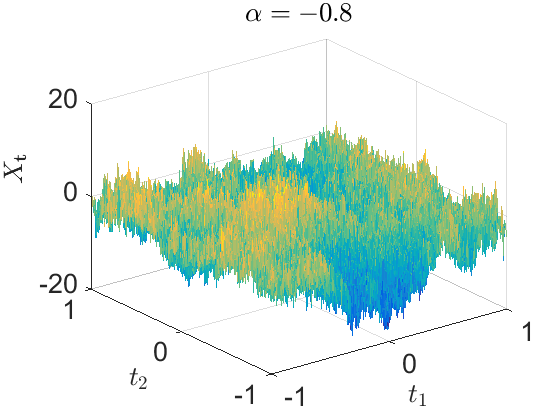}
 \includegraphics[width=0.48 \textwidth, height = 0.2 \textheight]{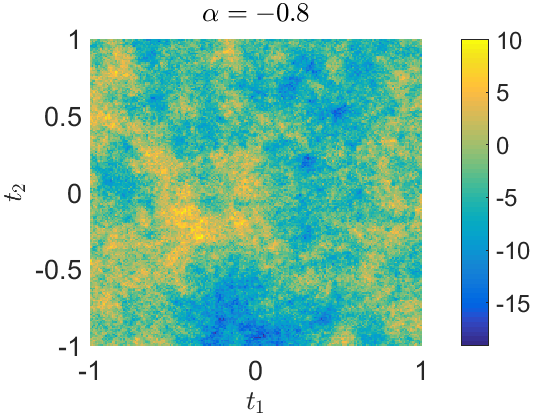}\\[2em]
\includegraphics[width=0.42 \textwidth, height = 0.2 \textheight]{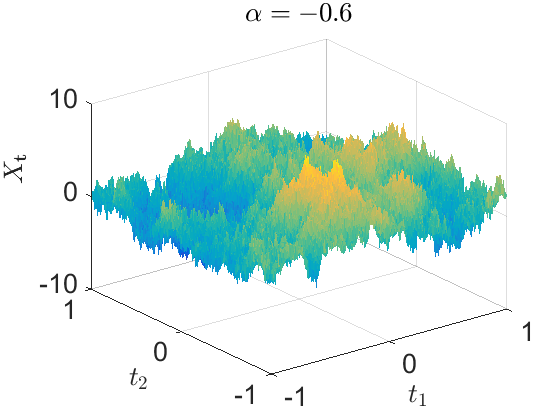}
 \includegraphics[width=0.48 \textwidth, height = 0.2 \textheight]{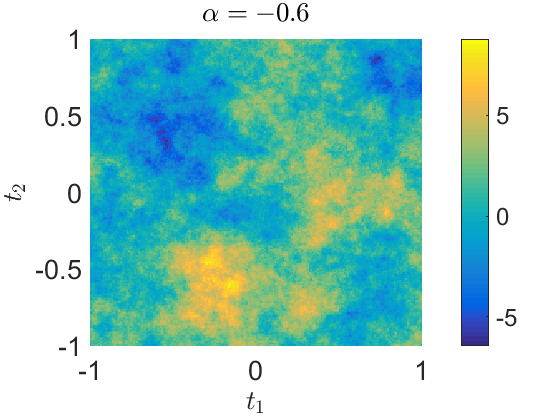}\\[2em]
 \includegraphics[width=0.42 \textwidth, height = 0.2 \textheight]{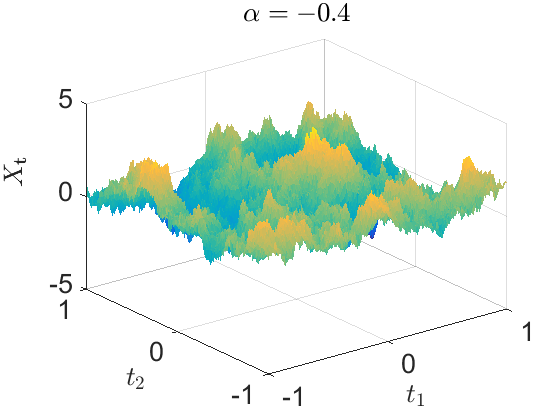}
 \includegraphics[width=0.48 \textwidth, height = 0.2 \textheight]{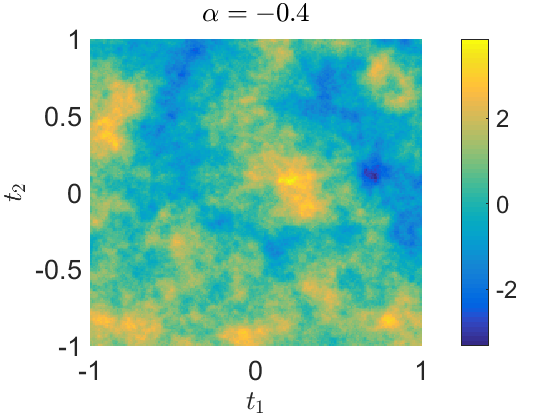}\\[2em]
 \includegraphics[width=0.42 \textwidth, height = 0.2 \textheight]{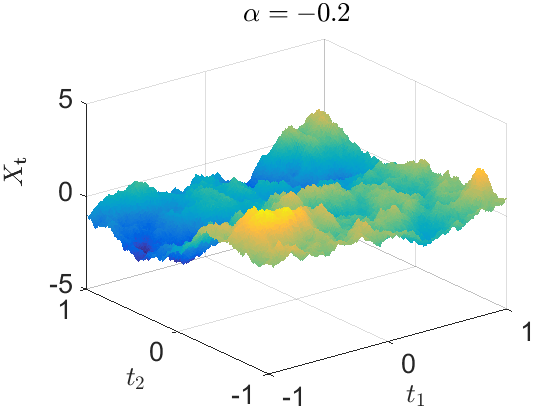}
 \includegraphics[width=0.48 \textwidth, height = 0.2 \textheight]{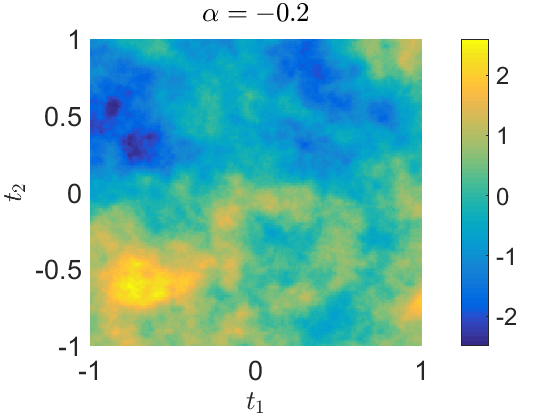}
  
\caption{Realisations of volatility modulated moving average fields for different $\alpha$ with Mat\'ern covariance, see Example \ref{Matern}. All plots range over $\bt\in[-1,1]^2$ and are generated with constant volatility $\sigma.$ In Section \ref{secNum} we present examples of VMMAs with nontrivial volatility.
 }
 \label{figSam}
  \end{figure}

The integral in \eqref{VMMA} is well defined, when $\sigma$ is measurable with respect to $\mathcal B(\R^2)\otimes \mathcal F$ and the process $\bss\mapsto g(\bt-\bss)\sigma_\bss(\omega)$ takes almost surely values in $L^2(\R^2)$. In particular we do not require independence of $\sigma$ and $W$ or any notion of filtration or predictability for the definition of the integral, as is usually used in the theory of stochastic processes indexed by time. This general theory of stochastic integration dates back to Bichteler \cite{Bic2002}, see also \cite{KwaWoy1992}. A brief discussion can be found in Appendix \ref{appInt}. When $\sigma$ and $W$ are independent, we can realise them on a product space and it is therefore sufficient to define integration with respect to $W$ for deterministic functions, which has been done in \cite{RajRos1989}.

The volatility field $(\sigma_\bss)_{\bss\in\R^2}$ is assumed to satisfy $\E[\sigma_\bss^2]<\infty$ for all $\bss$. Moreover, we assume $\sigma$ to be covariance stationary, meaning that $\E[\sigma_{\bss}]$ does not depend on $\bss$ and $\Cov(\sigma_{\bss+\br},\sigma_{\bss})=\Cov(\sigma_{\br},\sigma_{0})$ for all $\bss,\br\in\R^2$. In particular $\E[\sigma_\bss^2]=\E[\sigma_0^2]$ for all $\bss\in\R^2$. For some of our theoretical results we will assume that $\sigma$ and $W$ are independent, however we show in Appendix \ref{appInt} that this is not required for the convergence of the hybrid scheme.
We make the assumption that $\sigma$ is sufficiently smooth such that freezing $\sigma$ over small blocks will cause an asymptotically negligible error in the simulation. It turns out that this is the case when $\sigma$ satisfies
\begin{align}\label{sigmasmoothness}
\E[|\sigma_0-\sigma_{\bu}|^2]= o(\|\bu\|^{2\alpha+2}),\quad \text{for }\bu\to 0.
\end{align}
When $\sigma$ is independent of the Gaussian noise $W$, the covariance stationarity of $\sigma$ implies that the process $X$ is itself covariance stationary and covariance isotropic in the sense that $\E[(X_{\bt+\bss}-X_{\bt})^2]$ depends only on $\|\bss\|$. If $\sigma$ is in fact stationary, $X$ is stationary and isotropic.

Moreover, we pose the following assumptions on our kernel function $g$. They ensure in particular that $g$ is square integrable, which together with covariance stationarity of $\sigma$ ensures the existence of the integral in \eqref{VMMA}.
\begin{enumerate}[label=(A\arabic*)]
\item \label{1con0}The slowly varying function $L$ is continuously differentiable and bounded away from 0 on any interval $(u,\sqrt 2]$ for $u>0.$
\item  \label{1con1} It holds that $\wt g(x)=\mathcal O(x^\beta),$ as $x\to\infty$, for some $\beta\in (-\infty,-1),$
\item \label{1con2} There is an $M>0$ such that $|\wt g'|$ is decreasing on $[M,\infty)$ and satisfies
\[\int_1^\infty  \wt g'(r)^2r\,\diff r<\infty.\]
\item \label{1con3} There is a $C>0$ such that $|L'(x)|<C(1+x^{-1})$ for all $x\in (0,1].$ 
 \end{enumerate}

An appealing feature of the VMMA model is its flexibility in modelling marginal distributions and covariance structure independently. 
Indeed, assuming that $\sigma$ is stationary and independent of $W$, the covariance structure of $X$ is entirely determined by the kernel $g$, whereas the marginal distribution of $X$ is a centered Gaussian variance mixture with conditional variance $\int_{\R^2} g(-\bss)^2 \sigma_{\bss}^2\diff \bss$, the distribution of which is governed by the distribution of $\sigma.$

The behaviour of the kernel at $0$ is determined by the exponent $\alpha$, whereas its behaviour away from 0, e.g. how quickly it decays at $\infty$, depends on the slowly varying function $L$. While the behaviour of $g$ at 0 determines local properties of the process $X$, like the roughness of realisations, the behaviour of $g$ away from 0 governs its global properties, e.g., whether it is long range dependent. Being able to independently choose $\alpha$ and $L$ allows us therefore to model local and global properties of the VMMA independently, which underlines the flexibility of the model. This separation of local and global properties, and the desire to capture both of them correctly, is one of our main motivations to use a hybrid simulation scheme.
We now formalise the statement that the roughness of $X$ is determined by the power $\alpha$. 
\begin{theo}\label{Roughness}
\begin{itemize}
\item [(i)]Assume independence of $\sigma$ and $W$. The variogram of $X$ defined as $V(h):=\E[(X_0-X_\bt)^2]$, where $h=\|\bt\|$, satisfies \[h^{-2-2\alpha} L(h)^2 V(h)\to 2\pi \E[\sigma_0^2]\int_{\R^2} \big(\|\bx+\be/2\|^\alpha - \|\bx-\be/2\|^\alpha \big)^2 \diff \bx    \quad\text{as }h\to 0,\]
where $\be$ is any vector with $\|\be\|=1.$ 
\item [(ii)] Assume additionally that the volatility is locally bounded in the sense that it satisfies $\sup_{\|\bss\|\leq M+1} \big\{\sigma_\bss^2\big\}<\infty$ almost surely, where $M$ is as in assumption \ref{1con2}. Then, for all $\eps>0,$ the process $X$ has a version with locally $\alpha+1-\eps$-H\"older continuous realisations.
 \end{itemize}
 \end{theo}
The proof can be found in Section \ref{secPro}. In \cite{HanTho2013} the authors analyse the variogram of a closely related model and derive similar results. 
 
We conclude this section by discussing examples of possible choices for kernel functions $g$ and volatility fields $\sigma$. 

\begin{ex}[Mat\'ern covariance]\label{Matern} 
Originally introduced in the context of tree population modelling in Swedish forests by Bertil Mat\'ern \cite{Mat1960}, the Mat\'ern covariance family has become popular in a variety of different fields such as meteorology, hydrology and machine learning. For an overview we refer to \cite{GutGne2006} and the references therein. 
It is characterised by the correlation function
\[C(\|\br\|)= \E[(X_\br-X_0)^2]/\E[X_0^2]= \frac{(\lambda \|\br\|)^\nu}{2^{\nu-1}\Gamma(\nu)}K_\nu(\lambda \|\br\|),\quad \br\in\R^2,\]
where $\nu>0$ is usually referred to as the shape parameter, while $\lambda>0$ is a scale parameter. Here, $K_\nu$ denotes the modified Bessel function of the second kind.
It has been shown in \cite{JonRonMouJen2013}, see also \cite{HanTho2013}, that the model \eqref{VMMA} has Mat\'ern correlation, when 
\[g(\bt)=\|\bt\|^{\frac{\nu-1}2}K_{\frac{\nu-1}2}(\lambda \|\bt\|),\]
provided $\sigma$ is independent of $W$ and covariance stationary.
When $\nu\in(0,1)$, the function $g$ satisfies our model assumptions  \ref{1con0}-\ref{1con3} with $\alpha=\nu-1$, as we argue next. The function
\[L(x)=x^{\frac{1-\nu}2}K_{\frac{\nu-1}2}(\lambda x)\]
is continuously differentiable on $(0,\infty)$. It holds that $\lim_{x\downarrow 0}L(x)= 2^{-\frac{\nu+1}2}\Gamma\big(\frac{\nu-1}2\big)$, see \cite[Eq. (9.6.9), p.375]{AbrSte1964}, which implies that $L$ is slowly varying at 0 and satisfies condition \ref{1con3}. Moreover, since $K_{\frac{\nu-1}2}(\lambda x)$ decays exponentially as $x\to\infty$, cf. \cite[ p.378]{AbrSte1964}, condition \ref{1con1} is satisfied for any $\beta<-1.$ Condition \ref{1con2} follows as well from the exponential decay together with the identity
\[\frac{\diff}{\diff x}(x^{\alpha/2}K_{\alpha/2}(x))=x^{\frac \alpha 2-1}K_{\frac \alpha 2 -1}(x).\]
\end{ex}

\begin{ex}[ambit fields]\label{ambit fields}
In a series of papers \cite{BarSch2007,BarSch2008} the authors proposed to model velocities of particles in turbulent flows by a class of spatio-temporal stochastic processes called {\it ambit fields}. Over the last years this model found manifold applications throughout various sciences, examples being \cite{BarBenVer2011,JenJonSchBar2006}.
 The VMMA model is a purely spatial analogue of an ambit field driven by white noise and can therefore be interpreted as a realisation of an ambit field at a fixed time $t$. In the framework of turbulence modeling, the squared volatility $\sigma_\bss^2$ has the physical interpretation of local energy dissipation and it has been argued in \cite{BarHedSchSzo2016} that it is natural to model $\sigma_\bss^2$ as (exponential of) an ambit field itself. A possible model for the volatility is therefore $\sigma_\bt^2=\exp(X'_\bt)$ where $X'$ is a volatility modulated moving average, independent of $W$. Applying Theorem \ref{Roughness}~(i) it is not difficult to see that this model satisfies assumption \eqref{sigmasmoothness} when the roughness parameter $\alpha'$ of $X'$ satisfies $\alpha'>\alpha.$
In its core, an ambit field is a stochastic integral driven by a L\'evy basis, which does not need to be Gaussian. A simulation of such integrals in the non-Gaussian case typically relies on a shot noise decomposition of the integral, as demonstrated in \cite{Ros1990}, see also \cite{CohLacLed2008}.
\end{ex}


\section{The Hybrid Scheme}
\label{secHS}


In this section we present the hybrid simulation scheme using the following notation.
For $r>0$ and $\bt=(t_1,t_2)\in\R^2$ we introduce the notation $\sq_r\bt$ for a square with side length $1/r$ centred at $\bt$, that is  $\sq_r\bt=\big[t_1-\frac{1}{2r},t_1+\frac{1}{2r}\big]\times\big[t_2-\frac{1}{2r},t_2+\frac{1}{2r}\big].$ We will suppress the index $r$ if it is 1, and will denote $\sq_r$ instead of $\sq_r 0.$ We simulate the process $X_\bt$ for $\bt\in [-1,1]^2$ on the square grid $\Gamma_n:=\big\{\frac 1 n (i,j),\ i,j\in\{-n,...,n\}\big\}$.

A first necessary step for approximating the integral \eqref{VMMA} is to truncate the range of integration, i.e.
\[X_\bt\approx \int_{\sq_{1/C}\bt} g(\bt-\bss)\sigma_\bss W(d \bss),\]
for some large $C>0$. 
To ensure convergence of the simulated process as $n\to\infty$, we increase the range of integration simultaneously with increasing the grid resolution $n$. We let therefore $C=C_n\approx n^\gamma$ for some $\gamma>0.$ More precisely, it proves to be convenient to choose $C_n= \frac{N_n+1/2}{n}$ with $N_n=[n^{1+\gamma}]$, where $[x]$ denotes the integer part of $x$. 

An intuitive approach to simulating the model \eqref{VMMA} is approximating the integrand on $\sq_{C_n^{-1}}\bt$ by freezing it over squares with side length $1/n$, i.e.
\begin{align}\label{Riemann}
X^{R,n}_\bt= \sum_{\bj\in \bt+\{-N_n,...,N_n\}^2} g(\bt-\bb_\bj/n)\sigma_{\bj/n}\int_{\sq_n\bj}W(d \bss),
\end{align}
where $\bb_\bj\in \sq \bj$ are evaluation points chosen such that $\bt-\bb_\bj/n\neq 0$ for all $\bt\in\Gamma_n$ and $\bj\in \Z^2.$
Indeed, $X^{R,n}_\bt$ can be simulated, assuming that the volatility $\sigma$ can be simulated on the square grid $\big\{\frac 1 n (i,j),\ i,j\in\Z\big\}$, since $\big\{\int_{\sq_n\bj}W(d \bss)\big\}_{\bj\in\Z^2}\overset{\text{\tiny i.i.d}}{\sim}\mathcal N\big(0,\frac 1 {n^2}\big)$. We will refer to this simulation method as Riemann-sum scheme. The authors of \cite{NguVer2017} use this technique to simulate volatility moving averages with bounded moving average kernel and demonstrate that it performs well in this setting.
In our framework, however, a crucial weakness of this approach is the inaccurate approximation of the kernel function $g$ around its singularity at $0$, which results in a poor recovery of the roughness of $X$. 

This weakness can be overcome by choosing a small $\kappa\in \N_0$ (typically, $\kappa\in\{0,1,2\}$) and approximating $g$ by a power kernel on $\frac 1 n [-\kappa-1/2,\kappa+1/2]^2.$
More specifically, denoting $K_\kappa=\{-\kappa,\dots,\kappa\}^2$ and $\overline K_\kappa=\{-N_n,\dots,N_n\}^2\setminus K_\kappa$, 
the hybrid scheme approximates $X_\bt$ by 
\begin{align}\label{Hyb1}
X_\bt^{n}:= &\sum_{\bj\in K_\kappa}\sigma_{\bt-\bj/n} L(\|\bb_{\bj}\|/n) \int_{\sq_{n} (\bt-\bj/n)}\|\bt-\bss\|^\alpha W(d\bss)\nonumber\\
&+\sum_{\bj\in \ol K_\kappa}\sigma_{\bt-\bj/n}g(\bb_\bj/n)\int_{\sq_{n} (\bt-\bj/n)} W(d\bss).
\end{align}
See Figure \ref{HSvis} for a visualisation.
\begin{figure}
 \includegraphics[width= \textwidth]{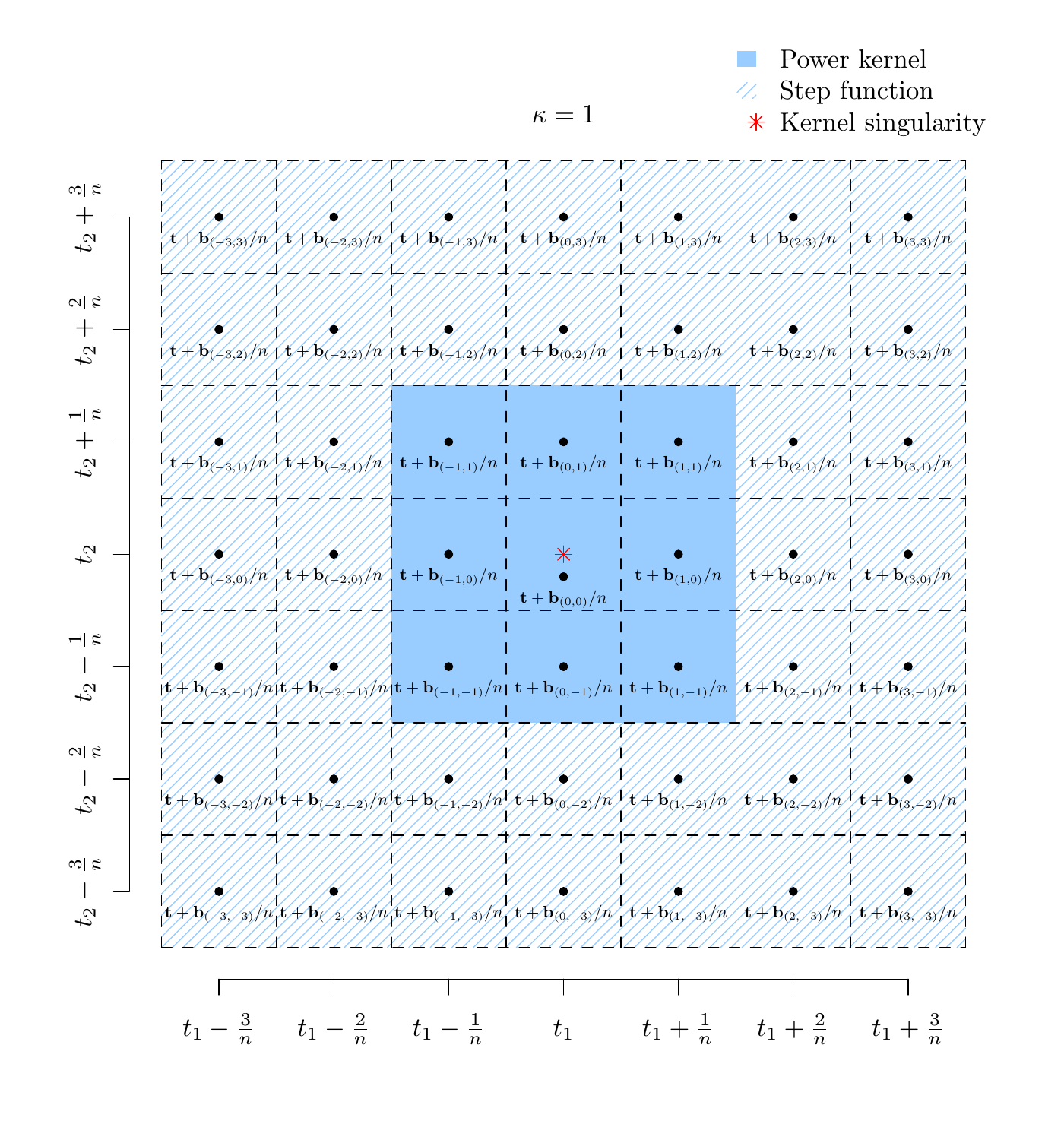}
 \caption{\label{HSvis}Visualisation of the hybrid scheme. Dividing $\R^2$ into small squares of size $1/n^2$, the kernel function $g$ is approximated by a power kernel in the squares close to the singularity, and by a step function further away. The figure shows the situation for $\kappa=1$, whereas for $\kappa=0$ ($\kappa=2$) the power kernel is used for only the central square (the central 25 squares). Simulating the random variables corresponding to the squares shown in the figure corresponds to simulating the process $X$ at $(t_1,t_2)\in\R^2$ only. For simulating $X$ at a different location $(t_1',t_2')$ we obtain the same pattern shifted and need to account for the covariances of the random variables (not shown).}
\end{figure}
In order to simulate $X_\bt $ on the grid $\bt\in\Gamma_n$, we simulate the families of centred Gaussian random variables $\mathcal W^1_n$ and $\mathcal W^2_n,$ defined as
\begin{align*}
\mathcal W^1_n&:=\bigg\{ W^{n}_{\bi,\bj}=\int_{\dn\bi/n}\|(\bi+\bj)/n -\bss\|^\alpha W(d \bss),\ W^n_\bi=\int_{\dn\bi/n} W(d \bss),\\
&\hspace{3em}\bi\in\{-n-\kappa,\dots, n+\kappa\}^2\text{ and }\bj\in K_\kappa
\bigg\},\\
\mathcal W^2_n&:=\bigg\{  W^n_\bi=\int_{\dn\bi/n} W(d \bss),\\
&\hspace{3em} \bi\in\{-N_n-n,\dots, N_n+n\}^2\setminus\{-n-\kappa,\dots, n+\kappa\}^2\bigg\}.
\end{align*}
 Indeed, replacing $\bt$ by $\bi/n$ in \eqref{Hyb1} yields
\begin{align}\label{X_i^nSim}
X^n_{\bi/n}&= 
\sum_{\bj\in K_\kappa} L (\|\bb_\bj\| ) \sigma_{\frac{\bi-\bj} n}W^n_{\bi-\bj,\bj}
+\sum_{\bj\in \overline K_\kappa} g(\bb_\bj /n)\sigma_{\frac{\bi-\bj} n} W^n_{\bi-\bj}\\
&=:\tilde X( \bi/ n)+\hat X( \bi/ n)
,\qquad\text{ for $\bi\in\{-n,\dots,n\}^2.$}\nonumber
\end{align}
By definition the random vectors $\big\{((W^{n}_{\bi,\bj})_{\bj\in K_\kappa},W^n_\bi),\bi\in\Z^2\big\}$ are independent and identically distributed along $\bi$. As a consequence, $\mathcal W^1_n$ and $\mathcal W^2_n$ are independent and $\mathcal W^2_n$ is composed of i.i.d.~$\mathcal N(0,1/n^2)$-distributed random variables. In order to simulate $\mathcal W^1_n$ we need to compute the covariance matrix of
$((W^{n}_{0,\bj})_{\bj\in K_\kappa},W^n_0)$, 
 which is of size $(|K_\kappa|+1)^2$ with $|K_\kappa|=(2\kappa+1)^2$. 
In contrast to the purely temporal model considered in \cite{BenLunPak2016}, computing the covariance structure becomes much more involved in our spatial setting. It relies partially on explicit expressions derived in Appendix \ref{appCov}, and partially on numeric integration.

Note that the complexity of computing $\tilde X(\frac \bi n)$ for all $\bi\in\{-n,...,n\}^2$ is $\mathcal O(n^2)$, as the number of summands does not increase with $n$. The sum $\hat X(\frac \bi n)$ can be written as the two dimensional discrete convolution of the matrices $A$ and $B$ defined by
\[A_\bk:= \begin{cases} 0&\bk\in K_\kappa\\
g(\bb_\bk/n)&\bk\in \ol K_\kappa \end{cases}
,\qquad B_\bk:= \sigma_{\bk/n} W^n_\bk ,\quad\text{for } \bk\in \{-N-n,...,N+n\}^2.
\]
We remark that this expression as convolution is the main motivation that in \eqref{Riemann} and \eqref{Hyb1} we chose to evaluate $\sigma$ at the midpoints $\bt-\bj/n$ of $\sq_n(\bt-\bj/n).$ 
Using FFT to carry out the convolution leads to a computational complexity of $\mathcal O(N^2 \log N)=\mathcal O(n^{2+2\gamma}\log n)$ for computing $\{\hat X(\frac \bi n)\}_{\bi\in\{-n,...,n\}^2}$. Consequently, the computational complexity of the hybrid scheme is $\mathcal O(n^{2+2\gamma}\log n)$, provided the computational complexity of simulating $\{\sigma_{\bi/n}\}_{\bi\in \{-N-n,...,N+n\}^2}$ does not exceed $\mathcal O(n^{2+2\gamma}\log n)$. By comparison we recall that the exact simulation of an isotropic Gaussian field using circulant embeddings is of complexity $\mathcal O(n^{2}\log n)$, see \cite{GneSevPerSchJia2006}. However, exact simulation requires $\sigma$ to be constant and the covariance structure to be known. 
If the kernel function $g$ is given, the covariance matrix often needs to be computed by numerical integration, leading to a complexity of $\mathcal O(n^4).$
 The complexity of Cholesky-factorisation for the covariance matrix of $X_\bt$ realised on the grid $\Gamma_n$ is $\mathcal O(n^6)$, see \cite[p.312]{AsmGly2007}.

Next we derive the asymptotics for the mean square error of the hybrid simulation scheme.
\begin{theo}\label{L2err}
Let $\alpha\in (-1,0)$. Assume that $\sigma$ is independent of $W$ and satisfies \eqref{sigmasmoothness}.
If $\gamma> -(1+\alpha)/(1+\beta)$, we have for all $\bt\in\R^2$ that
\[n^{2(\alpha+1)}L(1/n)^{-2}\E[|X_\bt-X^n_\bt|^2]  \to \E[\sigma_0^2]J(\alpha,\kappa,\bb),\quad \text{as }n\to\infty.\]
Here the constant $J(\alpha,\kappa,\bb)$ is defined as
\[J(\alpha,\kappa,\bb)=\sum_{\bj\in \Z^2\setminus\{-\kappa,\dots,\kappa\}^2}\int_{\boxempty\bj}(\|\bx\|^\alpha-\|\bb_{\bj}\|^\alpha)^2\diff\bx,\]
which is finite for $\alpha<0.$
\end{theo}

The proof is given in Section \ref{secPro}. This theorem and the computational complexity $\mathcal O(n^{2+2\gamma}\log n)$ of the hybrid scheme provide guidance how to choose the cutoff parameter $\gamma$. It should be chosen small under the constraint $\gamma> -(1+\alpha)/(1+\beta)$, where $\beta$ is chosen minimally such that \ref{1con1} is satisfied. For example if $X$ is of Mat\'ern type as in Example \ref{Matern}, the function $\wt g$ decays exponentially, and $\beta$ can be chosen arbitrarily small. In this case the asymptotic of the mean square error given in the theorem applies for any $\gamma>0.$

The sequence of evaluation points $\bb=(\bb_{\bj})_{\bj\in\Z^2}$ can be chosen optimally, such that it minimises the limiting constant $J(\alpha,\kappa,\bb)$ and thus the asymptotic mean square error of the hybrid scheme. To this end $\bb_{\bj}$ needs to be chosen in such a way that it minimises
 \[\int_{\sq \bj}(\|\bx\|^\alpha-\|\bb_{\bj}\|^\alpha)^2 d\bx,\]
 for all $\bj\in\Z^2$.
 By standard $L^2$ theory, $c\in\R$ minimises $\int_{\sq \bj}(\|\bx\|^\alpha-c)^2 d\bx$ if and only if the function $\bx\mapsto \|\bx\|^\alpha-c$ is orthogonal to constant functions, that is, if it satisfies
 \[\int_{\sq \bj}(\|\bx\|^\alpha-c) d\bx=0.\]
It follows then that $J(\alpha,\kappa,\bb)$ becomes minimal if we choose $\bb$ such that 
\begin{align}
\|\bb_{\bj}\|=\bigg(\int_{\sq \bj}\|\bx\|^\alpha d\bx\bigg)^{1/\alpha}.\label{optdis}
\end{align}
 In Appendix \ref{appCov}, we derive an explicit expression for this integral involving the Gau\ss\ hyperbolic function $_2F_1$. 
However, in our numerical experiments computing these integrals explicitly for all $\bj\in \ol K_\kappa$ slowed the hybrid scheme down considerably, and we recommend choosing the midpoints $\bb_\bj=\bj$ instead. Figure \ref{figJ} shows the constant $J(\alpha,\kappa,\bb^{\mathrm {opt}})=J_{\text{opt}}$ for optimally chosen evaluation points $\bb^{\mathrm{opt}}$ and the error caused by choosing midpoints $\bb_\bj=\bj$ instead, giving evidence that choosing midpoints leads to a nearly optimal result.

 \begin{figure}
 \includegraphics[width=0.45 \textwidth]{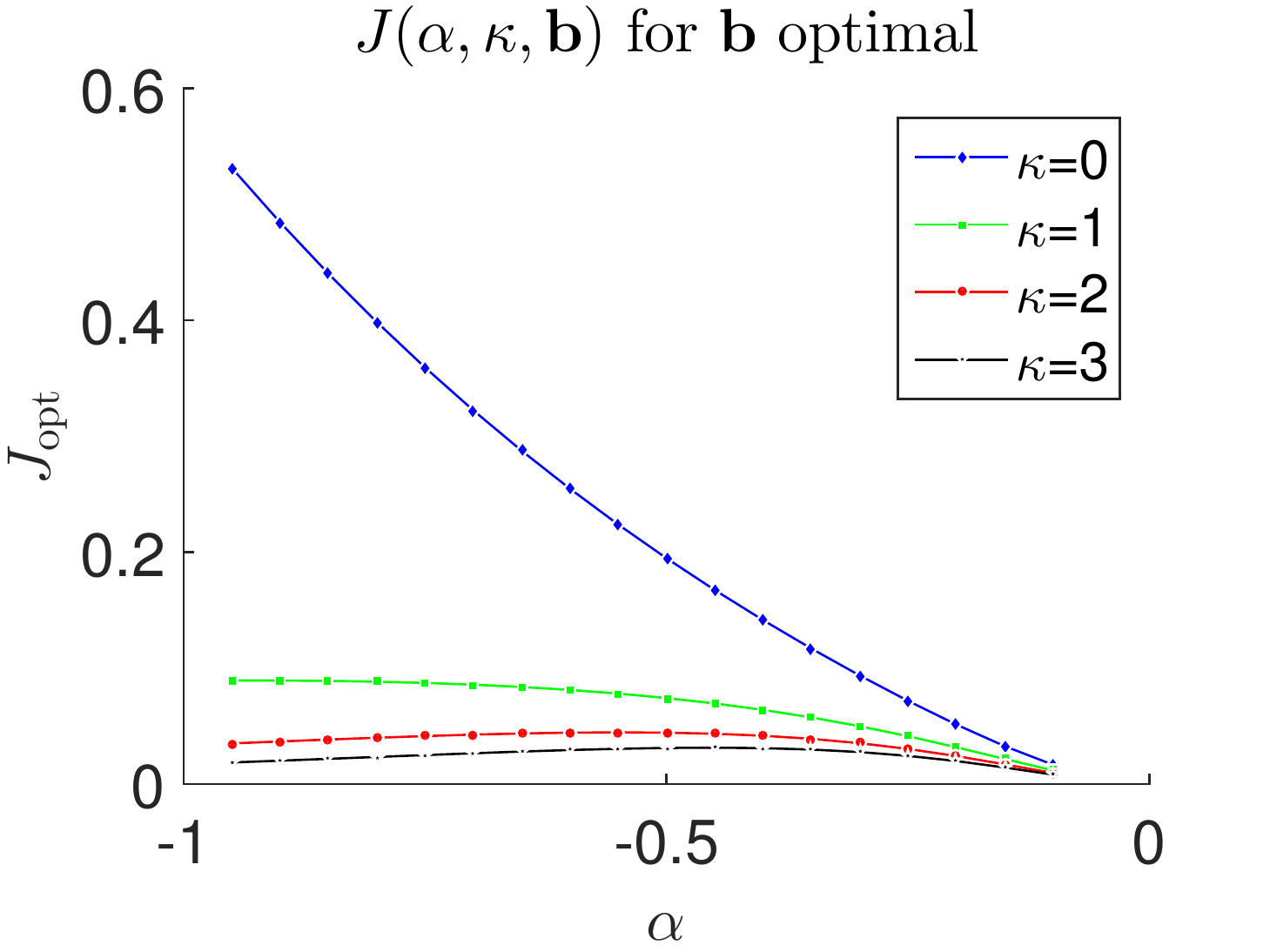}\hspace{1em}
 \includegraphics[width=0.45 \textwidth]{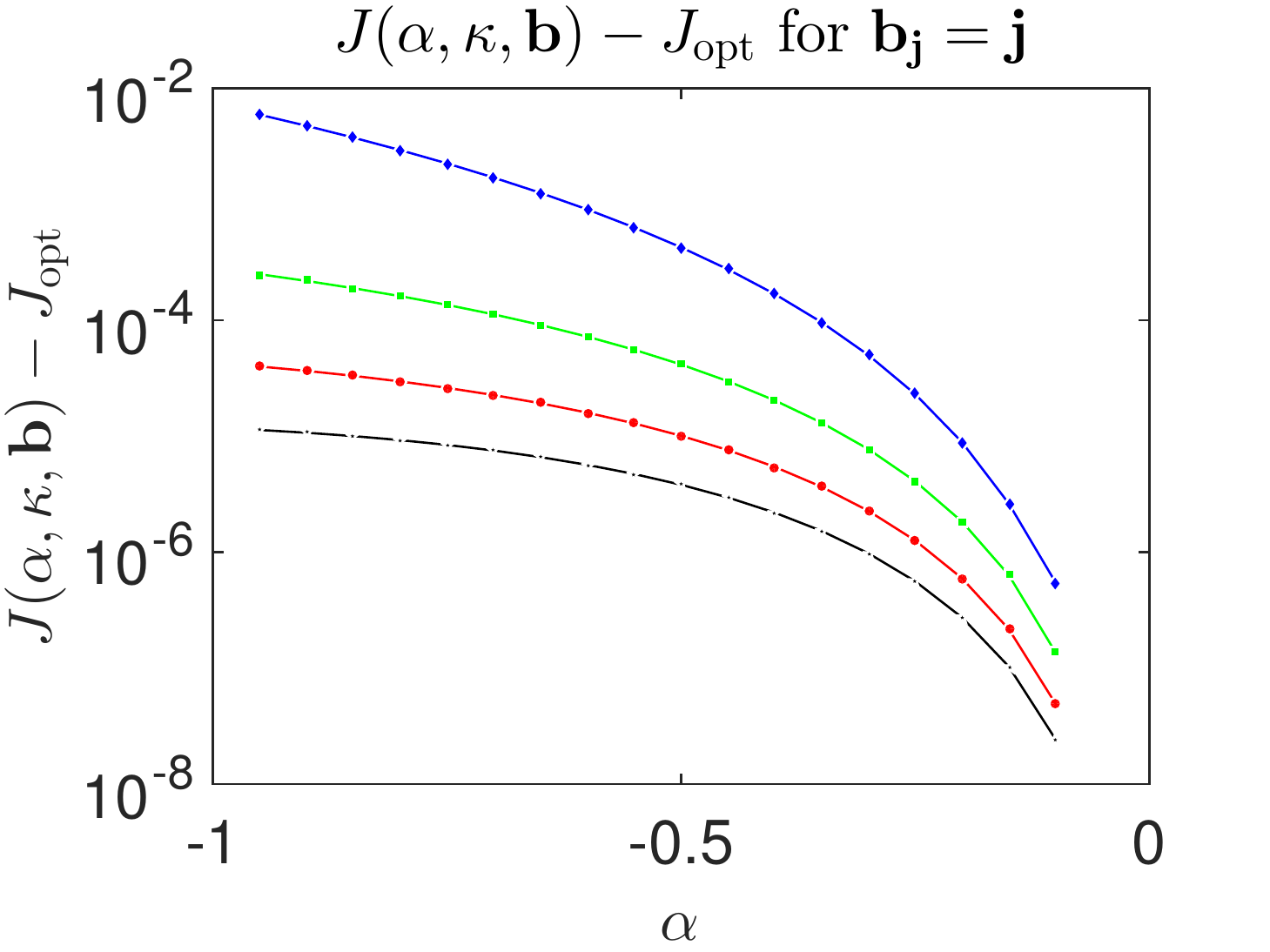}
 \caption{The first figure shows the value of $J(\alpha,\kappa,\bb^{\mathrm {opt}})=J_{\text{opt}}$ for different values of $\alpha$ and $\kappa$ where the evaluation points $\bb^{\mathrm {opt}}$ are chosen optimally, as in \eqref{optdis}. The second figure shows the absolute error $J(\alpha,\kappa,\bb)-J_{\text{opt}}$ for $\bb$ chosen as midpoints, i.e. $\bb_\bj=\bj$, demonstrating that this choice leads to close to optimal results.}\label{figJ}
  \end{figure}


For $\bj\in K_\kappa\setminus\{0\},$ the evaluation points $\bb_\bj$ do not appear in the limiting expression in Theorem \ref{L2err}, and we will simply choose the midpoints $\bb_\bj=\bj.$
However, for $\bj=0$ the expression $L(\|\bj\|)$ is not necessarily defined. Indeed, the slowly varying function $L$ might have a singularity at 0. This shows that particular attention should be paid to the choice of $\bb_0$, which is optimal if it minimises the $L^2$ error of the central cell, i.e.,
\[\bb_0=\argmin_{\bb\in\sq_n\setminus{\{ 0}\}}\E\bigg(\int_{\sq_n}g(\bss)W(d\bss)-L(\|\bb\| )\int_{\sq_n}\|\bss\|^\alpha W(d\bss)\bigg)^2. \]
By straightforward calculation it can be shown that this is equivalent to
\begin{align*}
L(\|\bb_0\|)&=\bigg(\int_{\sq_n}\|\bss\|^{2\alpha}L(\|\bss\|)\diff \bss\bigg)\bigg (\int_{\sq_n}\|\bss\|^{2\alpha}\diff \bss\bigg)^{-1}\\
 &= 8 C_{0,0}^{-1}\int_0^{1/\sqrt 2} r^{2\alpha+1}L(r/n) \big(\pi/4 - \arccos (\sqrt 2  r)\mathds 1_{\{r>1/2\}}\big)\diff r ,
 \end{align*}
where $C_{0,0}$ is defined in Appendix \ref{appCov}. The integral on the right hand side is finite for $\alpha>-1$, which follows from the Potter bound \eqref{PotBou}, and can be evaluated numerically.

Let us briefly mention that in principle the hybrid scheme can be extended to simulate stochastic processes of VMMA type in higher dimensions. However, 
to the best of our knowledge there are no closed form expressions for the covariance structure of the higher dimensional analogues of the Gaussian family $\mathcal W_n^1$ available, and they would need to be computed numerically.
Moreover, a similar scheme can be implemented that does not rely on the specific form of $g$ specified in \eqref{gAss} and does therefore in particular allow for anisotropic fields, when the covariance matrix of $\mathcal W_n^1$ is computed numerically. More specifically, replacing $L(\|\bb_\bj\|)W^n_{\bi-\bj,\bj}$ in the definition of $\wt X(\bi/n)$ in \eqref{X_i^nSim} by 
\[\wt W^n_{\bi-\bj,\bj}:=\int_{\dn(\bi-\bj)/n}g(\bi/n -\bss) W(d \bss),\]
the covariance matrix of the i.i.d.~random vectors $\big\{((\wt W^{n}_{\bi,\bj})_{\bj\in K_\kappa},W^n_\bi),\bi\in\Z^2\big\}$ can be computed by numerical integration. Thereafter, $X^n_{\bi/n}$ can be simulated as in \eqref{X_i^nSim}. An obvious drawback of this approach, apart from being more computationally involved, is that in this general setup the roughness of the random field \eqref{VMMA} cannot be characterised by a single parameter $\alpha$, and we do not pursue this idea further.


\section{Numerical results}
\label{secNum}

In this section we demonstrate in a simulation study that the hybrid scheme is capable of capturing the roughness of the process correctly, and compare it in that aspect to other simulation schemes. Before doing so, we present in Figure~\ref{figVol} samples of VMMAs highlighting the effect of volatility. The volatility is modelled as $\sigma_\bt^2=\exp(X'_\bt),$ where $X'$ is again a volatility modulated moving average, compare Example~\ref{ambit fields}. For $X'$ we choose the roughness parameter $\alpha=-0.2$ and the slowly varying function $L(x)=e^{-x}$. For the first realisation we chose $\alpha=-0.3$ and $L(x)=e^{-x}$. For the second we chose $\alpha=-0.7$ and $L$ such that the model has Mat\'ern covariance, see Example~\ref{Matern}. In both cases it becomes apparent that areas of lower volatility cause the VMMA field to vary less.

 \begin{figure}
 \includegraphics[width=0.47 \textwidth]{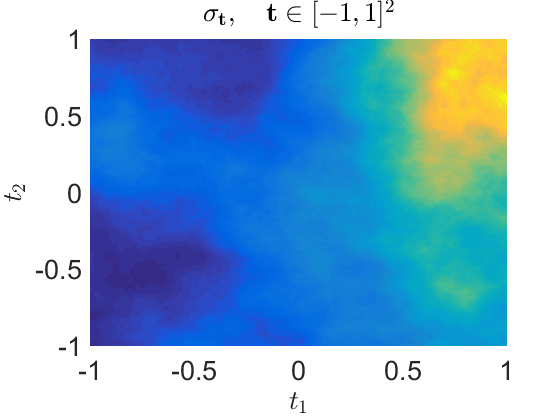}\hspace{1em}
 \includegraphics[width=0.47 \textwidth]{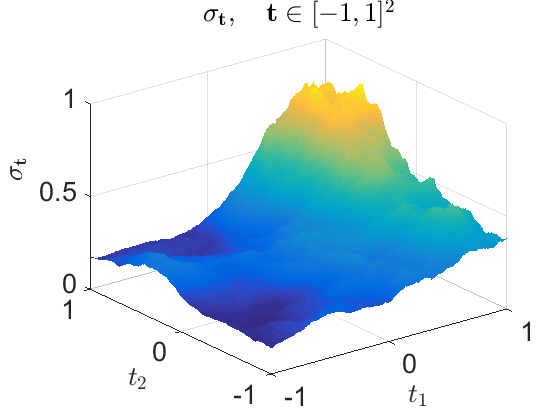}\\[2em]
 \includegraphics[width=0.47 \textwidth]{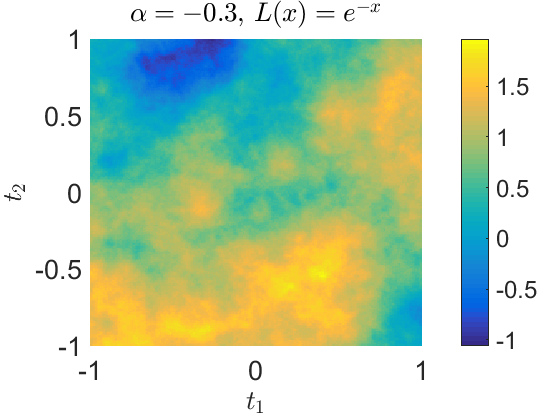}\hspace{1em}
 \includegraphics[width=0.47 \textwidth]{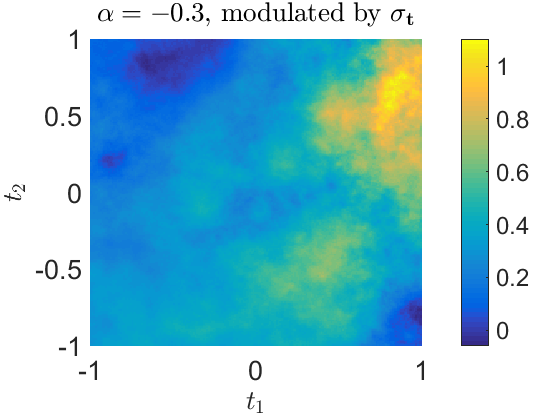}\\[2em]
 \includegraphics[width=0.47 \textwidth]{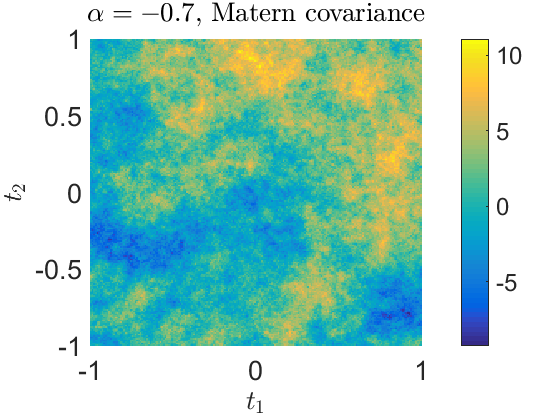}\hspace{1em}
 \includegraphics[width=0.47 \textwidth]{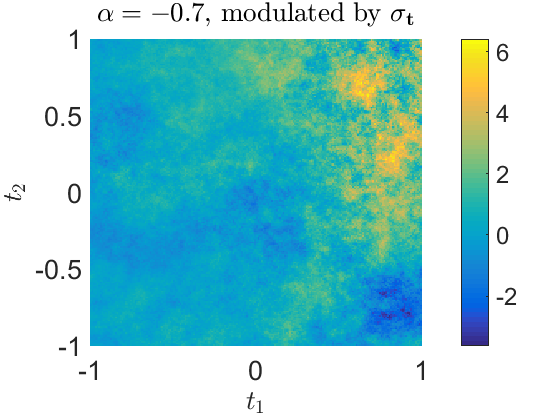}
  \caption{Examples for moving average fields modulated by volatility. The first row shows the volatility $(\sigma_\bt)_{\bt\in\R^2}$ modelled as $\sigma^2_\bt=\exp(X'_\bt)$, where $X'$ is again a VMMA field.
The second and third row show realisations of VMMAs. On the left hand side the field is simulated with constant volatility, the right hand side is generated by the same Gaussian noise and with the same model parameters, but is modulated by $(\sigma_{\bt})_{\bt\in\R^2}$.
For the second row we chose $\alpha=-0.3$ and the slowly varying function $L(x)=e^{-x}.$ The third row is generated with $\alpha=-0.7$ and Mat\'ern covariance.}
 \label{figVol}
  \end{figure}

For our simulation study we first recall the definition of fractal or Hausdorff dimension. For a set $S\subset \R^d$ and $\eps>0$, an $\eps$-cover of $S$ is a countable collection of balls $\{B_i\}_{i\in\N}$ with diameter $|B_i|\leq \eps$ such that $S\subset \bigcup_{i} B_i.$ The $\delta$-dimensional Hausdorff measure of $S$ is then defined as 
\[H^\delta (S)=\lim_{\eps\to 0}\inf\bigg\{\sum_{i=1}^\infty |B_i|^\delta\ :\ \{B_i\}_{i\in\N}\text{ is }\eps\text{-cover of }S\bigg\},\]
and the fractal or Hausdorff dimension of $S$ is 
\[\HD(S):=\inf\{\delta>0\ :\ H^{\delta}(S)=0\}.\]
The Hausdorff dimension of a spatial stochastic process $(X_\bt)_{\bt\in\R^2}$ is the (random) Hausdorff dimension of its graph $\HD(\{(\bt,X_{\bt}),\ \bt\in \R^2\})$, and takes consequently values in $[2,3].$ For the model \eqref{VMMA} with constant volatility $\sigma\equiv 1$ it follows easily from a standard result \cite[Theorem 8.4.1]{Adl1981} and Theorem \ref{Roughness} that $\HD(X)=2-\alpha,$ see also \cite{HanTho2013}. 
In \cite{GneSevPer2012}, the authors give an overview over existing methods for estimating the Hausdorff dimension of both time series data and spatial data, and provide implementations for various estimators in form of the \texttt{R} package \texttt{fractaldim} \cite{fractaldim}, which we rely on.

We estimate the Hausdorff dimension from simulations of $X$ generated by the hybrid scheme, and compare to estimates from other simulation methods.
We consider the model \eqref{VMMA} with constant volatility $\sigma$ and Mat\'ern covariance, see Example \ref{Matern}. In this case the process $X$ can be simulated exactly using circulant embeddings of the covariance matrix, to which we compare.
Note that exact simulation is only available for Gaussian processes with known covariance function and is not applicable for general VMMAs. Moreover we compare to the Riemann-sum scheme introduced in \eqref{Riemann}. For the hybrid scheme we consider $\kappa=0,1,2,3$.
With each technique we simulate 100 i.i.d.~Monte-Carlo samples of the process $(X_\bt)_{\bt\in[-1,1]^2}$ for every $\alpha\in \{-0.8,-0.7,...,-0.1\}$. As grid resolution we chose $n=100$ and, for the hybrid scheme and the Riemann-sum scheme, $N_n=[n^{1+\gamma}]$ with $\gamma=0.3,$ i.e. $N_n=398.$
Thereafter we estimate the roughness of $X$ and average the estimates over the Monte-Carlo samples. There is a variety of different estimators for fractal dimension of spatial data.
For a detailed overview and asymptotic properties we refer to \cite{GneSevPer2012} and the references therein. We apply the square increment estimator $\nu_{SI}$ introduced and analysed by Chan and Wood \cite[(4.3)]{ChaWoo2000} because of its favourable asymptotic properties, see \cite{ChaWoo2000,GneSevPer2012}. Figure \ref{RoughEstFig} shows the results and compares them to the theoretical value of the Hausdorff dimension $2-\alpha$, plotted as dashed line. For the second plot in the figure we remark that the sample variance of the roughness estimates was between 0.005 and 0.01, for all values of $\alpha$ and all simulation methods.

Exact simulation using circulant embeddings performs slightly better than the hybrid scheme, in particular when $\alpha\approx 0.$ This is not surprising, taking into account that the roughness of the process is governed by the behaviour of the kernel $g$ at 0, which is well approximated by the hybrid scheme but, intuitively speaking, perfectly recovered by exact simulation. Let us stress again that exact simulation using circulant embeddings is only available for the model \eqref{VMMA} in a few special cases. For $\kappa\geq 1$ the hybrid simulation scheme recovers the roughness very precisely, when $\alpha<-0.3.$ When $\alpha\geq -0.3$ or $\kappa=0$ it still performs reasonably well but tends to overestimate the roughness of the process slightly. This behaviour is likely to be caused by the at 0 slowly varying function, $L(x)=x^{-\alpha/2}K_{\alpha/2}(x)$ in the Mat\'ern covariance case, which, intuitively speaking, varies more at 0 for larger values of $\alpha.$ As expected, the Riemann-sum approximation underestimates the roughness of the field significantly, as it does not account for the explosive behaviour of $g$ at 0.

For the exact simulation via circulant embeddings we used the \texttt{R} package \texttt{RandomFields} \cite{RandomFields}, and refer to \cite{GneSevPerSchJia2006} for more details on this simulation method. For the roughness estimation we relied on the \texttt{R} package \texttt{fractaldim} \cite{fractaldim}. Our implementation of the hybrid scheme is in MATLAB.

In Table \ref{ComTim} we compare computation times for the hybrid scheme, the circulant embeddings method, and the Riemann-sum scheme.
For generating a single realisation, the circulant embedding method and the Riemann-sum scheme perform faster than the hybrid scheme. The main reason for this, however, is the costly computation of the covariance of the family $\mathcal W^1_n$, which is only required once when generating i.i.d.~Monte-Carlo samples. In view of the rather long computation times for all algorithms, let us stress that $n=100$ corresponds to simulating $X$ on a fine grid containing $(2n+1)^2=40,401$ grid points. 

\begin{table}
\begin{tabular}{lcccccc}
MC samples    & $\kappa= 0$ & $\kappa= 1$ & $\kappa= 2$ & $\kappa=3$ & circ.emb.& Riemann-sum\\[.5ex]
\hline\\[-1ex]
1     & 12.6 s   & 13.2 s	& 14.3 s	&	15.3 s	& 0.8 s    &  1.2 s\\
100   & 51 s     & 61.3 s	& 72.6 s	& 	77.7 s	&75.6 s    &  32.5 s   

\end{tabular}
\caption{Computation time of the hybrid scheme for different $\kappa$, for exact simulation using circulant embeddings, and for the Riemann-sum scheme, for a Mat\'ern covariance Gaussian field.
The first row shows the computation time for a single realisation, the second for 100 i.i.d.~samples.
The parameters of the model were chosen as $n=100,\ \alpha=-0.6,$ and, for the hybrid and the Riemann-sum scheme, $\gamma=0.3 $. The computation time was measured on a computer with with 2.9 GHz CPU and 32 GB RAM.  }\label{ComTim}
\end{table}


\begin{centering}
 \begin{figure}
 \includegraphics[width=0.9 \textwidth]{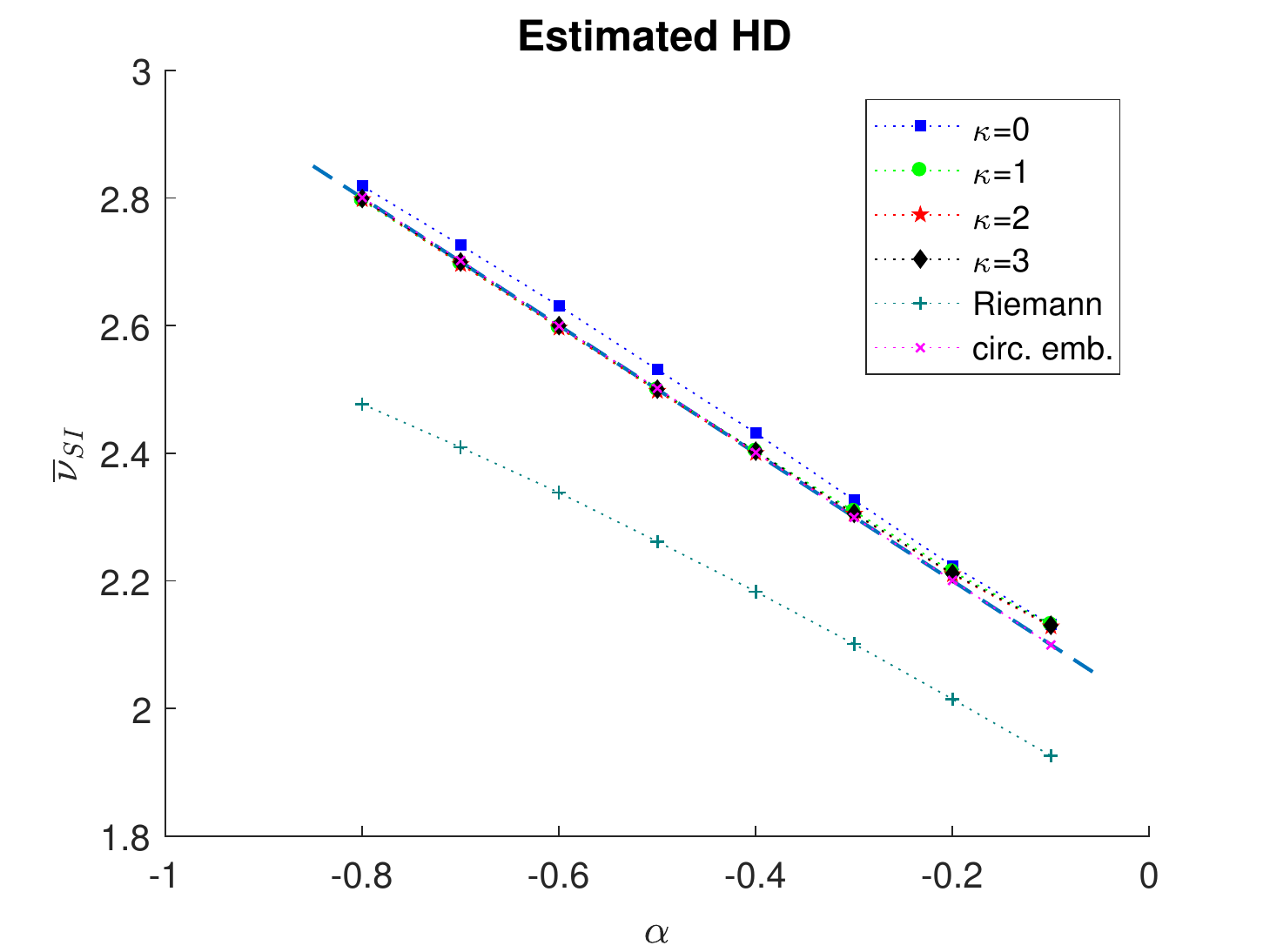}\\[2em]
 \includegraphics[width=0.9 \textwidth]{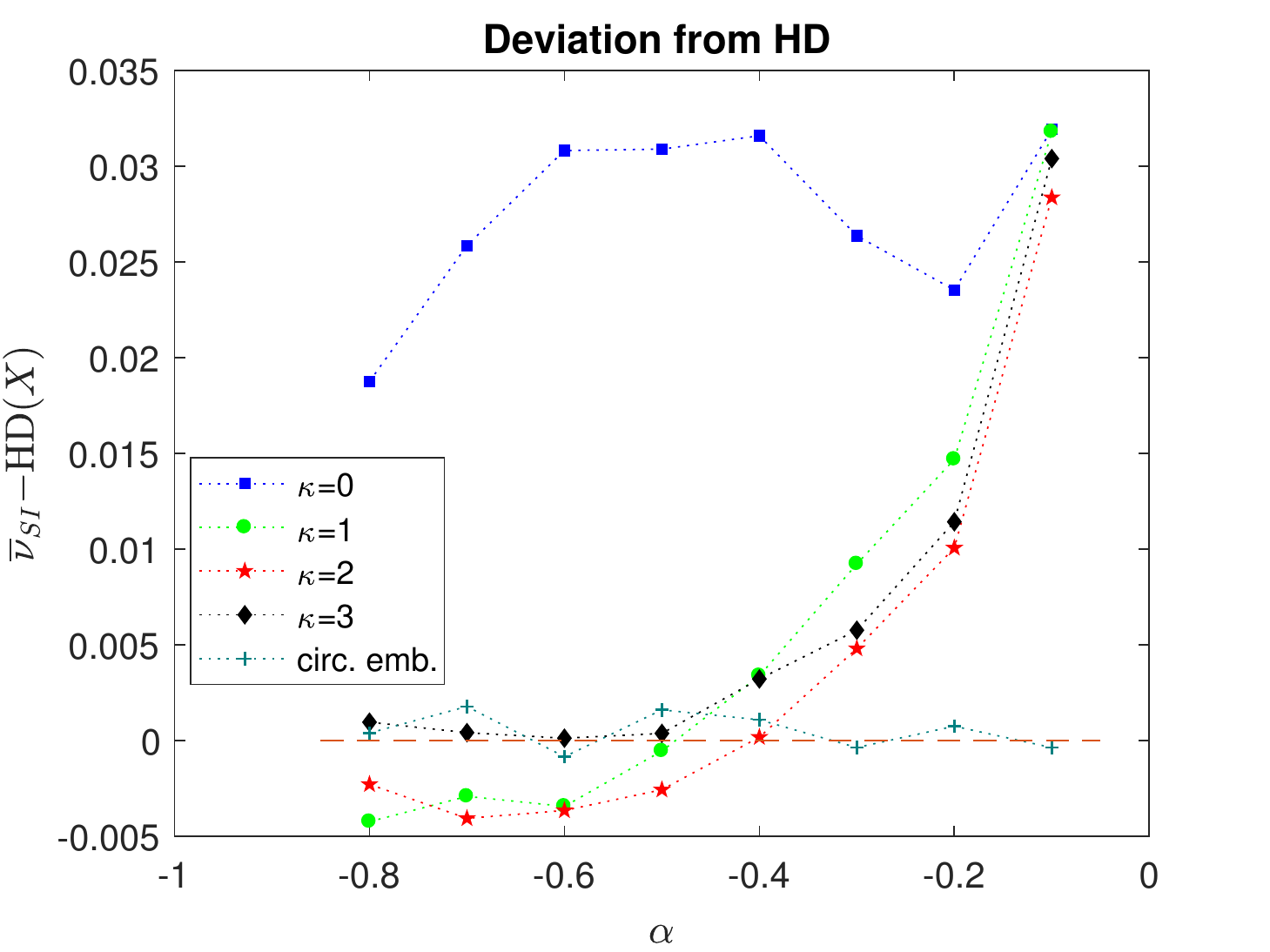}
 \caption{Roughness estimated from samples generated by the hybrid scheme, the Riemann-sum approximation method and by exact simulation using the circulant embedding method for Gaussian fields. The theoretical roughness is marked as a dashed line. The roughness is estimated by the isotropic estimator $\nu_{SI}$ introduced in \cite{ChaWoo2000}, averaged over 100 i.i.d.~samples.  The second plot shows in more detail the deviation between the estimation and the theoretical value, not including the Riemann-sum approximation scheme.}
 \label{RoughEstFig}
 \end{figure}
 \end{centering}

\section{Proofs}\label{secPro}

This section is dedicated to the proofs of our theoretical results. We begin by recalling the Potter bound which follows from \cite[Theorem 1.5.6]{BinGolTeu1989}. For any $\delta>0$ there exists a constant $C_\delta>0$ such that
\begin{align}\label{PotBou}
L(x)/L(y)\leq C_\delta\max\bigg\{\bigg(\frac{x}{y}\bigg)^\delta,\bigg(\frac{x}{y}\bigg)^{-\delta}\bigg\},\quad x,y\in(0,1].
\end{align}
This bound will play an important role throughout all the proofs in this section.

\begin{proof}[Proof of Theorem \ref{Roughness} (i)]
The proof is similar to the proof of \cite[Proposition 2.1]{BenLunPak2016}.
 We have for $h>0$ by covariance stationarity of $\sigma$ that
 \[V(h)=\E[\sigma_0^2] \int_{\R^2}\big(g(\bss+h\be)-g(\bss)\big)^2 \diff\bss, 
\]
where $\be$ is any unit vector and we used transformation into polar coordinates. We obtain
\begin{align*}
V(h) &= \E[\sigma_0^2] (A_h+A_h'), \quad \text{where }\\
 A_h&=\int_{\{\|\bss\|\leq 1 \}} \big( g(\bss+h\be/2)- g(\bss-h\be/2))^2 \diff \bss,\quad\text{ and }\\
 A'_h&=\int_{\{\|\bss\|> 1 \}} \big( g(\bss+h\be/2)- g(\bss-h\be/2))^2 \diff \bss.
  \end{align*}
Since the function $\wt g$ is continuously differentiable on $(0,\infty)$, we obtain by the mean value theorem the following estimate for $A'_h$.
\begin{align*}
A'_h\leq h^2 \bigg\{ &\int_{\{1<\|\bss\|<M+1 \}} \ \sup_{\{\xi\,:\,|\xi-\|\bss\||\leq  h/\sqrt 2\}} (\wt g'(\xi))^2 \ \diff \bss\\
&+ 2\pi \int_{M}^\infty \wt g'(r)^2  r\diff r\bigg\},
\end{align*}
where we used that $|\wt g'|$ is decreasing on $[M,\infty).$ The term in curly brackets is finite by Assumption \ref{1con2}, and we obtain that $A'_h=\mathcal O(h^2)$, as $h\to 0.$ 
For $A_h$ we make the substitution $\bx=\bss/h$ and obtain
\begin{align*}
 A_h&=h^2\int_{\|\bx\|\leq 1/h} \big( g(h(\bx+\be/2))-g(h(\bx-\be/2)))^2 \diff \bx \\
&=h^{2+2\alpha} L^2(h)\int_{\|\bx\|\leq 1/h} G_h(\bx)\diff \bx,
\end{align*}
where
\[G_h(\bx)=\bigg(\|\bx+\be/2\|^\alpha \frac{L(h\|\bx+\be/2\|)}{L(h)}-\|\bx-\be/2\|^\alpha \frac{L(h\|\bx-\be/2\|)}{L(h)}\bigg)^2 .\]
Note that $G_h(\bx)\to \big(\|\bx+\be/2\|^\alpha- \|\bx-\be/2\|^\alpha\big)^2$, as $h\to 0$. Therefore the first statement of the theorem follows by the dominated convergence theorem if there is an integrable function $G$ satisfying $G(\bx)\geq |G_h(\bx)|$ for all $\bx$ for sufficiently small $h$. The existence of such a function follows since $L$ is bounded away from 0 on $(0,1]$ and by Assumption \ref{1con3}. For details we refer to the proof of \cite[Proposition 2.2]{BenLunPak2016}. 
 \end{proof}
 
 \begin{proof}[Proof of Theorem \ref{Roughness} (ii)]
  The proof relies on the Kolmogorov-Chentsov theorem (cf.\ \cite[Theorem 3.23]{Kal2002}), which requires localisation of the process, as $\sigma$ does not necessarily have sufficiently high moments. We therefore first show the existence of a H\"older continuous version under the assumption that there is an $m>0$ such that
\begin{align}
\label{Loc1}&|\sigma_\bss|^2\leq m,\hspace{4em}\text{ for all }\bss\text{ with $\|\bss\|\leq M+1,\ \omega\in\Omega$, and} \\
&\int_{\{\|\bss\|\geq M+1\}} (g(\bt-\bss)-g(-\bss))^2\sigma_\bss^2 \diff\bss\leq m\|\bt\|^2\nonumber,\\
\label{Loc2}&\hspace{8.5em}\text{ for all }\bt\text{ with }\|\bt\|\leq 1,\ \omega\in\Omega,
\end{align}
where $M$ is as in \ref{1con2}. Thereafter we argue that the theorem remains valid if we relax these assumptions to $\E[\sup_{\|\bss\|\leq M}\sigma_\bss^2]<\infty.$

For $\|\bt\|\leq 1$ we have for all $p>0$ that 
\begin{align*}
\E[(X_\bt-X_0)^p]&\leq C_{p}\E\bigg[\bigg(\int_{\R^2}\big(g(\bt-\bss)-g(-\bss)\big)^2\sigma_{\bss}^2\diff \bss\bigg)^{p/2}\bigg]\\
&\leq C_p m^{p/2} \bigg(\int_{\{\|\bss\|\leq M+1\}}\big(g(\bt-\bss)-g(-\bss)\big)^2\diff \bss + \|\bt\|^2\bigg)^{p/2}\\
&\leq C_p m^{p/2} \bigg( V_0(\|\bt\|) + \|\bt\|^2\bigg)^{p/2},
\end{align*} 
where $V_0$ denotes the variogram of the process $(X_\bt)_{\bt\in\R^2}$ with $\sigma\equiv 1.$ In the first inequality we used that $\sigma$ and $W$ are independent and therefore $X_\bt-X_0$ has a Gaussian mixture distribution with the integral on the right hand side being the conditional variance. Applying the first part of the theorem and the Potter bound \eqref{PotBou} we obtain that for any $\delta>0$ a constant $C_{p,m,\delta}$ such that for all $\bt$ with $\|\bt\|\leq 1$
\[\E[(X_\bt-X_0)^p]\leq C_{p,m,\delta} \|\bt\|^{p+p\alpha-\delta}.\]
Therefore, the Kolmogorov-Chentsov Theorem \cite[Theorem 2.23]{Kal2002} implies that $X$ has a continuous version that is H\"older continuous of any order $\gamma<1+\alpha-\frac \delta p -\frac 2 p,$ and the result follows for any $\gamma\in(0,1+\alpha)$ by letting $p\to\infty.$

We will now complete the proof of the theorem by extending it to processes not satisfying assumptions \eqref{Loc1} and \eqref{Loc2}.
By mean value theorem we obtain that for all $\bt$ with $\|\bt\|\leq 1$
\begin{align*}
&\|\bt\|^{-2 }\int_{\{\|\bss\|\geq M+1\}} (g(\bt-\bss)-g(-\bss))^2\sigma_\bss^2 \diff\bss\\
 &\eqspace\leq \|\bt\|^{-2 }\int_{\{\|\bss\|\geq M+1\}} |\|\bt-\bss\|-\|\bss\||^2\sup_{r\in[\|\bss\|,\|\bt-\bss\|]}\big(\wt g'(r)^2\big)\sigma_\bss^2 \diff\bss\\
 &\eqspace\leq \int_{\{\|\bss\|\geq M+1\}} \wt g'(\|\bss\|-1)^2\sigma_\bss^2 \diff\bss ,
\end{align*}
where we used that $|\wt g'|$ is decreasing on $[M,\infty)$. By taking expectation and transformation into polar coordinates it follows from Assumption \ref{1con2} that the right hand side is almost surely finite. Consequently, the random variable
\[Z:=\max\bigg\{\sup_{\|\bss\|\leq M+1}\big(\sigma_\bss^2\big),\sup_{\|\bt\|\leq 1}\bigg(\|\bt\|^{-2 }\int_{\{\|\bss\|\geq M+1\}} (g(\bt-\bss)-g(-\bss))^2\sigma_\bss^2 \diff\bss\bigg)\bigg\}\]
is almost surely finite. The process $(X_\bt \mathds 1_{\{Z\leq m\}})_{\bt\in \R^2}$ satisfies conditions \eqref{Loc1} and \eqref{Loc2} and coincides with $X$ on $\{Z\leq m\}$. Therefore, the existence of a version of $X$ with $\alpha+1-\eps$-H\"older continuous sample paths follows by letting $m\to\infty.$
\end{proof}

For the proof of Theorem \ref{L2err} we need the following auxiliary result. The proof is similar to the proof of \cite[Lemma 4.2]{BenLunPak2016} and not repeated. 

\begin{lem}\label{SlowVar}Let $\alpha\in\R$ and $\bj\in \Z^2\setminus\{(0,0)\}$. If $\bb_{\bj}\in \sq \bj,$ it holds that
\begin{enumerate}[label=(\roman{enumi})]
\item $\displaystyle\lim_{n\to\infty}\int _{\boxempty \bj}\bigg(\|\bx\|^\alpha \frac{L(\|\bx\|/n)}{L(1/n)}-\|\bb_\bj\|\frac{L(\|\bb_\bj\|/n)}{L(1/n)}\bigg)^2 \diff\bx =\int_{\boxempty\bj}(\|\bx\|^\alpha-\|\bb\|^\alpha)^2\diff\bx,$
\item $\displaystyle\lim_{n\to\infty}\int _{\boxempty \bj}\|\bx\|^{2\alpha}\bigg(\frac{L(\|\bx\|/n)}{L(1/n)}-\frac{L(\|\bb_\bj\|/n)}{L(1/n)}\bigg)^2 \diff\bx
	=0.$
\end{enumerate}
The same holds for $\bj=(0,0)$ if $\bb_{(0,0)}\neq (0,0)$ and $\alpha>-1.$
 \end{lem}

\begin{proof}[Proof of Theorem \ref{L2err}]
Recall the definition 
\begin{align*}
X_\bt^{n}:= &\sum_{\bj\in K_\kappa} \int_{\sq_{n} (\bt-\bj/n)}\|\bt-\bss\|^\alpha L(\|\bb_\bj\|)\sigma_{\bt-\bj/n}W(d\bss)\\
&+\sum_{\bj\in \ol K_\kappa}\int_{\sq_{n} (\bt-\bj/n)} g(\bb_\bj/n)\sigma_{\bt-\bj/n}W(d\bss).
\end{align*}

We introduce the auxiliary object ${X'}^{n}$ defined as 
\begin{align*}
{X'_\bt}^{n}:= &\sum_{\bj\in K_\kappa\cup \ol K_\kappa }\sigma_{\bt-\bj/n} \int_{\sq_{n} (\bt-\bj/n)}g(\bt-\bss)W(d\bss)\\
&+\int_{\R^2\setminus \sq_{N_n/n}\bt} g(\bt-\bss)\sigma_{\bss}W(d\bss).
\end{align*}
Denoting $E_n:=\E[|X^n_\bt-{X'}^n_\bt|^2]$ and $E_n':=\E[|X_\bt-{X'}^n_\bt|^2],$ Minkowski's inequality yields 
\begin{align}\label{sigmaestimate}
E_n(1-\sqrt{E_n'/E_n})^2\leq \E[|X^n_\bt-X_\bt|^2]\leq E_n(1+\sqrt{E_n'/E_n})^2.
\end{align}
We will show later that $E_n'/E_n\to 0$ as $n\to\infty,$ and it is thus sufficient to analyse the asymptotic behaviour of $E_n.$

We have that
\begin{align}\label{Edecomp}
E_n
=&\sum_{\bj\in K_\kappa} \int_{\sq_n(\bt-\bj/n)} \big(\|\bt-\bss\|^\alpha L(\|\bb_\bj\|/n)-g(\bt-\bss)\big)^2 \E[\sigma_{\bt-\bj/n}^2]\diff \bss\nonumber\\
&+\sum_{\bj\in\{-n,\dots,n\}^2\setminus K_\kappa}\int_{\sq_n(\bt-\bj/n)} \big(g(\bt-\bss)-g(\bb_{\bj}/n)\big)^2 \E[\sigma_{\bt-\bj/n}^2] \diff \bss\nonumber\\
&+\sum_{\bj\in\ol K_\kappa \setminus \{-n,\dots,n\}^2}\int_{\sq_n(\bt-\bj)} \big(g(\bt-\bss)-g(\bb_{\bj}/n)\big)^2 \E[\sigma_{\bt-\bj/n}^2]\diff \bss\nonumber\\
&+\int_{\R^2\setminus \sq_{(2N_n+1)/n}\bt}  g(\bt-\bss)^2 \E[\sigma_\bss^2]\diff \bss\nonumber\\
=&\E[\sigma_0^2](D_1+D_2+D_3+D_4).
\end{align}
For $D_4$ we obtain, recalling assumption \ref{1con1} and $N_n=n^{\gamma+1}$ that 
\[D_4\leq \int_{\|\bss\|>N_n/n} g(\bss)^2 d\bss=\mathcal O((N_n/n)^{2\beta+2})= \mathcal O(n^{2\gamma(1+\beta)}).\]
Therefore, we have 
\begin{align}\label{D4}
n^{2(1+\alpha)}D_4\to 0.
\end{align}
For $D_3$ we obtain
\[D_3=\sum_{\bj\in\ol K_\kappa \setminus \{-n,\dots,n\}^2}\int_{\sq_n\bj/n} \big(g(\bss)-g(\bb_{\bj}/n)\big)^2 \diff \bss.\]
Recalling the notation $\tilde g(\|\bss\|)=g(\bss)$ we have for $\bss\in\boxempty\bj$ with $\bj\in\ol K_\kappa \setminus \{-n,\dots,n\}^2$ by the mean value theorem $\xi\in[\|\bss\|\wedge\|\bb_{\bj}/n\|,\|\bss\|\vee\|\bb_{\bj}/n\|]$. Since $\tilde g'$ is decreasing on $[M,\infty)$ by assumption \ref{1con2} it follows that
\begin{align*}
|g(\bss)-g(\bb_{\bj}/n)|&=| \tilde g'(\xi )(\|\bss\|-\|\bb_{\bj}\|/n)|\\
&\leq 
\begin{cases}
\frac 1 n \sup_{y\in [1-1/(\sqrt 2n), M+1/(\sqrt 2n)]}|\tilde g'(y)|,&(\|\bj\|-\sqrt 2)/n < M,\\
\frac 1 n |\tilde g'((\|\bj\|-\sqrt 2)/n)|,&(\|\bj\|-\sqrt 2)/n \geq M.
\end{cases}
\end{align*}
  Consequently, we obtain with transformation into polar coordinates
\begin{align}\label{D3}
&\limsup_{n\to\infty}n^2D_3\\
&\eqspace\leq \bigg(\pi(M+1)^2\sup_{z\in[1/2,M+1/2]}|\tilde g'(z)|+C\int_M^\infty r|\tilde g'(r)|^2\diff r\bigg)<\infty.
\end{align}

For $D_1$ we have that
\begin{align*}
D_1&=\frac 1 {n^2}\sum_{\bj\in K_\kappa} \int_{\boxempty\bj} \big(\|\bss/n\|^\alpha L(\|\bb_\bj\|/n)-g(\bss/n)\big)^2 \diff \bss\\
&=\frac {L(1/n)} {n^{2+2\alpha}}\sum_{\bj\in K_\kappa} \int_{\boxempty\bj} \|\bss\|^{2\alpha} \bigg(\frac{L(\|\bb_\bj\|/n)}{L(1/n)}-\frac{L(\|\bss\|/n)}{L(1/n)}\bigg)^2 \diff \bss.
\end{align*}
Since the number of elements of $K_\kappa$ does not depend on $n,$ we have by Lemma \ref{SlowVar}
\begin{align}\label{D1}
\lim_{n\to\infty}\frac{n^{2+2\alpha}D_1}{L(1/n)}=0.
\end{align}

The term $D_2$ can be written as
\begin{align*}
D_2&=\frac 1 {n^2}\sum_{\bj\in\{-n,\dots,n\}^2\setminus K_\kappa} \int_{\boxempty \bj}\big(g(\bss/n)-g(\bb_\bj/n)\big)^2\diff \bss\\
&=\frac {L(1/n)^2} {n^{2+2\alpha}}\sum_{\bj\in\{-n,\dots,n\}^2\setminus K_\kappa} \underbrace{\int_{\boxempty \bj}\bigg(\|\bss\|^\alpha \frac{L(\|\bss\|/n)}{L(1/n)}-\|\bb_\bj\|^\alpha \frac{L(\|\bb_\bj\|/n)}{L(1/n)}\bigg)^2\diff \bss}_{:= A_{\bj,n}}.
\end{align*}
From Lemma \ref{SlowVar} we know that $\lim_{n\to\infty}A_{\bj,n}=\int_{\boxempty\bj}(\|\bss\|^\alpha-\|\bb_\bj\|^\alpha)^2\diff\bss.$ Consequently, if we find a dominating sequence $A_\bj$ such that $A_\bj\geq A_{\bj,n}$ for all $n$ and $\sum_{\bj\in\Z^2\setminus K_\kappa} A_\bj<\infty,$ it follows from dominated convergence theorem that
\begin{align}\label{D2negalpha}
\hspace{-1em}\lim_{n\to\infty} \frac{D_2 n^{2\alpha+2}}{L(1/n)^2}=\sum_{\bj\in \Z^2\setminus K_\kappa} \int_{\boxempty\bj}(\|\bss\|^\alpha-\|\bb_\bj\|^\alpha)^2\diff\bss,\quad\text{for }\alpha\in(-1,0).
\end{align}
It holds that
\begin{align*}
A_{\bj,n}&= \int_{\boxempty\bj}\bigg\{ \big(\|\bss\|^\alpha-\|\bb_\bj\|^\alpha\big) \frac{L(\|\bss\|/n)}{L(1/n)} +\|\bb_\bj\|^\alpha \bigg(\frac{L(\|\bss\|/n)}{L(1/n)}- \frac{L(\|\bb_\bj\|/n)}{L(1/n)}\bigg)\bigg\}^2\diff\bss\\
		&\leq 2\int_{\boxempty\bj} (\|\bss\|^\alpha -\|\bb_\bj\|^\alpha)^2\bigg( \frac{L(\|\bss\|/n)}{L(1/n)}\bigg)^2\diff\bss\\
		&\hspace{1em}+2\int_{\boxempty\bj} \|\bb_\bj\|^{2\alpha}\bigg( \frac{L(\|\bss\|/n)-L(\|\bb_\bj\|/n)}{L(1/n)}\bigg)^2\diff\bss\\
&=:I_{\bj,n}+I'_{\bj,n}.
\end{align*}
For $I'_{\bj,n}$ we note that $\|\bb_\bj\|^{2\alpha}\leq (\|\bj\|-1/\sqrt 2)^{2\alpha}$ for $\alpha<0$. By the mean value theorem we have a $\xi\in[\|\bss\|/n\wedge\|\bb_\bj\|/n,\|\bss\|/n\vee\|\bb_\bj\|/n]$ such that
\[|L(\|\bss\|/n)-L(\|\bb_\bj\|/n)|=L'(\xi)|\|\bss\|/n-\|\bb_\bj\|/n|\leq \frac {C}{n}+\frac{C}{\|\bj\|-1/\sqrt 2}\leq\frac{2C}{\|\bj\|-1/\sqrt 2},
\] 
where we used \ref{1con3} and that $\|\bj\|\leq n.$ Consequently, we obtain
\begin{align}
I'_{\bj,n}&\leq \frac {C}{\inf_{x\in(0,1]}L(x)}(\|\bj\|-1/\sqrt 2)^{2\alpha}\int_{\boxempty \bj} (L(\|\bss\|/n)-L(\|\bb_\bj\|))^2\diff \bss\\
&\leq C( \|\bj\|-1/\sqrt 2)^{2(\alpha-1)}.
\end{align}
For the term $I_{\bj,n}$ we obtain by the Potter bound and the mean value theorem that
\[I_{\bj,n}\leq C_\delta\int_{\boxempty \bj}\min(\|\bss\|,b_{\bj})^{2\alpha-2}\|\bss\|^{2\delta}\diff\bss\leq C_\delta (\|\bj\|-1/\sqrt 2)^{2(\alpha-1+\delta)},\]
where we choose $\delta\in(0,-\alpha).$
Consequently, we obtain $I_{\bj,n}+I'_{\bj,n}\leq C (\|\bj\|-1/\sqrt 2)^{-2}$ for all $n>0$, and since
\[\sum_{\bj\in\Z^2\setminus K_\kappa}C(\|\bj\|-1/\sqrt 2)^{-2}<\infty, \]
\eqref{D2negalpha} follows from dominated convergence theorem and Lemma \ref{SlowVar}. Now \eqref{Edecomp} together with \eqref{D4}, \eqref{D3}, \eqref{D1} and \eqref{D2negalpha} show that
\[E_n\sim \E[\sigma_0^2]J(\alpha,\kappa,\bb)n^{-2(\alpha+1)}L(1/n)^2,\quad n\to\infty.\]
Therefore, recalling \eqref{sigmaestimate}, the proof of statement (i) of the Theorem can be completed by showing that $E_n'/E_n\to 0$ as $n\to \infty.$

Since $\sigma$ is covariance stationary, we obtain for $E_n'$
\begin{align*}
E_n'&=\sum_{\bj\in K_\kappa\cup \ol K_\kappa } \int_{\sq_{n} (\bt-\bj/n)}\E[(\sigma_{\bt-\bj/n}-\sigma_\bss)^2]g(\bt-\bss)^2 d\bss\\
	&=\sup_{\bu\in\sq_n} \E[|\sigma_\bu-\sigma_0|^2]\int_{\R^2} g(\bss)^2 d\bss,
\end{align*}
and $E_n'/E_n\to 0$ follows by the assumption \eqref{sigmasmoothness}
\end{proof}

\appendix
\section{On general stochastic integrals}\label{appInt}


We recall the definition of general stochastic integrals of the form $\int_{\R^2}H_\bss W(d\bss)$ where $H$ is a real valued stochastic process, not necessarily independent of $W.$ 
The construction of such integrals dates back to Bichteler \cite{Bic2002}. 
  In a recent publication \cite{ChoKlu2015}, this theory is revisited in a spatio-temporal setting and the authors derive a general integrability criterion for stochastic integrals driven by a random measure that is easy to check. In the context of integrals of the form \eqref{VMMA}, this criterion yields the following statement. 
\begin{prop}
Let $(H_\bss)_{\bss\in\R^2}$ be a real valued stochastic process, measurable with respect to $\mathcal B(\R^2)\otimes \mathcal F$, such that $H\in L^2(\R^2)$, almost surely. Then, the stochastic integral
$\int_{\R^2}H_\bss W(d\bss)$ exists in the sense of \cite{Bic2002}.
\end{prop}
\begin{proof}
We apply the integrability criterion \cite[Theorem 4.1]{ChoKlu2015} that is formulated in a spatio-temporal framework. To this end, we introduce an artificial time component and lift the white noise $W(d\bss)$ to a space time white noise $\wt W(dt;d\bss)$ such that $W(A)=\wt W([0,1]\times A)$ for all $A\in\mathcal B(\R^2).$ Equipping $(\Omega,\mathcal F,\P)$ with the maximal filtration $\mathcal F_t=\mathcal F$ for all $t\in[0,1],$ the spatio-temporal process defined as $H_{\bss}(t):=H_\bss$ for all $t\in[0,1]$ is predictable and it holds that
\[\int_{\R^2} H_\bss W(d\bss)=\int_{[0,1]\times\R^2} H_\bss(t) \wt W(dt;d\bss)\]
if the latter exists.
The random measure $\wt W$ satisfies the conditions of \cite[Theorem 4.1]{ChoKlu2015} with characteristics $B=\mu=\nu=0$ and $C(A;B)=\lambda(A\cap B)$ for all $A,B\in \mathcal B([0,1]\times \R^2)$, where $\lambda$ denotes the Lebesgue measure. The theorem then implies that $H$ is integrable with respect to $W$ if and only if it satisfies almost surely $\int_{\R^2} H_{\bss}^2 d\bss<\infty.$
\end{proof}

Note that the proofs for some of our theoretical results rely on the isometry
\[\E\bigg[\bigg(\int_{\R^2}H_\bss W(d\bss)\bigg)^2\bigg]=\E\bigg[\int_{\R^2} H_\bss^2 d\bss  \bigg],\]
which does not necessarily hold when $H$ and $W$ are dependent. In particular, we cannot rely on Theorem \ref{L2err} in this more general framework. We argue next that the hybrid scheme converges for dependent $\sigma$ and $W$, when $\sigma$ admits a continuous version, without specifying the speed of convergence.

\begin{prop}
Assume that $(\sigma_\bss)_{\bss\in\R^2}$ has a continuous version. Then, $X_\bt^n\toop X_\bt$ for all $\bt\in\R^2$, i.e. the hybrid scheme converges.
\end{prop}
\begin{proof}
Using the notation of Section \ref{secHS}, we consider the auxiliary integrals
\[\wt X^n_\bt:=\sum_{\bk\in K_\kappa\cup \ol K_\kappa}\sigma_{\bt-\bk/n}\int_{\sq_{n} (\bt-\bk/n)}g(\bt-\bss)W(d\bss)=\int_{\R^2}\wt \sigma^n_\bss g(\bt-\bss)W(d\bss),\]
where
\[\wt\sigma^n_\bss:=\sum_{\bk\in K_\kappa\cup \ol K_\kappa}\sigma_{\bt-\bk/n}\mathds 1_{\sq_{n} (\bt-\bk/n)}(\bss).\]
By arguing as in the proof of Theorem \ref{L2err}, it follows that $\E[(\wt X^n_\bt-X^n_\bt)^2]\to 0$ as $n\to \infty$, and it is therefore sufficient to argue that $\wt X^n_\bt\toop X_\bt.$ It holds that
\[X_\bt= \int_{\R^2} g(\bt-\bss) \sigma_{\bss}W(d\bss)=\int_{\R^2}\sigma_{\bss} M_{g,\bt}(d \bss),\]
where the random measure $M_{g,\bt}$ is defined as $M_{g,\bt}(A)=\int_A g(\bt-\bss) W(d\bss).$ Since $(\sigma_\bss)_{\bss\in\R^2}$ is continuous, the sequence of simple integrands $\wt\sigma^n$ converges pointwise to $\sigma$, and it follows that
\[X_\bt=\int_{\R^2} \sigma_\bss M_{g,\bt}(d \bss)= \lim_{n\to\infty}\int_{\R^2} \wt\sigma^n_\bss M_{g,\bt}(d \bss)=\lim_{n\to\infty}\wt X^n_\bt,\quad \text{in probability,}\]
by integrability of $\sigma$ with respect to $M_{g,\bt}$.
\end{proof}

\section{The covariance of $\mathcal W^1_n$}
\label{appCov}

In this section we analyse the covariance structure of the Gaussian family $\mathcal W^1_n$ introduced in Section \ref{secHS}. For a wide range of covariances we are able to derive closed expressions, whereas the remaining covariances are computed by numerical integration. Let us remark that, in addition to the symmetry of the covariance matrix, the isotropy of the process adds 8 more spatial symmetries (corresponding to the linear transformations in the orthogonal group $O(2)$ that map the grid $\Gamma_n$ onto itself), which reduces the number of necessary computations drastically.
Since the random variables in $\mathcal W^1_n$ are i.i.d.~along $\bi$, it is sufficient to derive the covariance matrix for 
\[\big\{ W^{n}_{0,\bj},\ W^n_0 \big\}_{\bj\in K_\kappa}.\]
For $\bj_1,\bj_2\in\{-\kappa,\dots,\kappa\}^2$ it holds that
\begin{align*}
C_{1,1}&:=\Var(W^n_0)=\frac 1 {n^2},\\
C_{1,\bj_1}&:=\Cov(W^n_0,W^{n}_{0,\bj_1})=\frac{1}{n^{2+\alpha}}\int_{\sq}\|\bj_1-\bss\|^\alpha\diff\bss,\\
C_{\bj_1,\bj_2}&:=\Cov(W^{n}_{0,\bj_1},W^{n}_{0,\bj_2})=\frac 1 {n^{2+2\alpha}}\int_{\sq}\|\bj_1-\bss\|^\alpha\|\bj_2-\bss\|^\alpha\diff\bss.
\end{align*}

We now derive explicit expressions for $C_{\bj,\bj}$ using the Gauss hypergeometric function $_2F_1$. Clearly, these expressions can be applied to compute $C_{1,\bj}$ by replacing $\alpha $ with $\alpha/2$.
Using symmetries we may assume without loss of generality that $\bj=(j_1,j_2)$ with $j_1\geq j_2\geq 0$. 
We introduce the notation $\tr\bj$ for the area $\{(x_1,x_2)\ :\ j_2\leq x_1\leq j_1,\ j_2\leq x_2\leq x_1\},$ that is a right triangle with lower right vertex $(j_1,j_2)$ and hypotenuse lying on the diagonal $\{(x_1,x_2)\,:\,x_1=x_2\}$. 
In order to obtain explicit expressions for $C_{\bj,\bj}$, we first derive explicit expressions for 
\begin{align}\label{triInt}
\int_{\tr\bj}\|\bx\|^{2\alpha}\diff\bx,\quad \text{for all }\bj=(j_1,j_2)\in\R^2,0\leq j_2< j_1.
\end{align}
Thereafter we give for all $\bj=(j_1,j_2)\in \Z^2$ with $0\leq j_2\leq j_2$ an explicit formula to write $C_{\bj,\bj}$ as linear combination of such integrals.

 Transforming into polar coordinates we obtain that
\begin{align}\label{CovCom}
\hspace{-3em}\int_{\tr \bj}\|\bx\|^{2\alpha} \diff \bx
	&=\int_{\arctan(j_2/j_1)}^{\pi/4} \int_{j_2/\sin(\theta)}^{j_1/\cos(\theta)}r^{2\alpha+1} \diff r \diff \theta\nonumber\\
	&=\frac 1{2\alpha+2} \int_{\arctan(j_2/j_1)}^{\pi/4}\bigg(\frac{j_1}{\cos(\theta)}\bigg)^{2\alpha+2}-\bigg(\frac{j_2}{\sin(\theta)}\bigg)^{2\alpha+2} \diff \theta.
\end{align}
It holds that $\arctan(j_2/j_1)=\arccos(\frac{j_1}{\|\bj\|}),$ and consequently we obtain by substituting $\cos(\theta)=z$ the following expression for the first summand:
\begin{align*}
& \frac {j_1^{2\alpha+2}}{2\alpha+2}\int_{\arctan(j_2/j_1)}^{\pi/4}\cos(\theta)^{-2\alpha-2}\diff \theta\\
	&\eqspace=-\frac {j_1^{2\alpha+2}}{2\alpha+2}\int_{j_1/\|\bj\|}^{\cos(\pi/4)} z^{-2\alpha-2}(1-z^2)^{-1/2}\diff z\\
	&\eqspace=\frac {j_1^{2\alpha+2}} {4(\alpha+1)}\int^{j_1^2/\|\bj\|^2}_{1/2} z^{-\alpha-\frac 3 2}(1-z)^{-1/2}\diff z\\
	&\eqspace=\frac {j_1^{2\alpha+2}} {4(\alpha+1)}\int_{j_2^2/\|\bj\|^2}^{1/2} (1-z)^{-\alpha-\frac 3 2}z^{-1/2}\diff z\\
	&\eqspace=\frac {j_1^{2\alpha+2}} {4(\alpha+1)} (B(1/2;1/2,-\alpha-1/2)-B(j_2^2/\|\bj\|^2;1/2,-\alpha-1/2))\\
	&\eqspace=\frac {j_1^{2\alpha+2}} {2^{3/2}(\alpha+1)} {}_2F_1(1/2,3/2+\alpha;3/2;1/2)\\
	&\eqspace\hspace{1em}-\frac{{j_1^{2\alpha+2}}j_2}{2\|\bj\|(\alpha+1)}{}_2F_1(1/2,3/2+\alpha;3/2;j_2^2/\|\bj\|^2).
\end{align*}
Here, $B(x;p,q)$ denotes the incomplete beta function, satisfying $B(x;p,q)= \frac{x^p} p {_2F_1}(p,1-q;p+1;x)$. For the first equality we used that $d/dz (\arccos(z))=-(1-z^2)^{-1/2}$
For the second summand in \eqref{CovCom} we argue similarly, using that $\arctan(j_2/j_1)=\arcsin(\frac{j_2}{\|\bj\|})$,
\begin{align*}
&- \frac {j_2^{2\alpha+2}}{2\alpha+2}\int_{\arctan(j_2/j_1)}^{\pi/4}\sin(\theta)^{-2\alpha-2}\diff \theta\\
	&\eqspace=- \frac {j_2^{2\alpha+2}}{2\alpha+2}\int_{j_2/\|\bj\|}^{\sin(\pi/4)} z^{-2\alpha-2}(1-z^2)^{-1/2}\diff z\\
	&\eqspace=- \frac {j_2^{2\alpha+2}}{4(\alpha+1)}\int_{j_2^2/\|\bj\|^2}^{1/2} z^{-\alpha-\frac 3 2}(1-z)^{-1/2}\diff z\\
	&\eqspace=- \frac {j_2^{2\alpha+2}}{4(\alpha+1)}\int^{j_1^2/\|\bj\|^2}_{1/2} (1-z)^{-\alpha-\frac 3 2}z^{-1/2}\diff z\\
	&\eqspace= - \frac {j_2^{2\alpha+2}}{4(\alpha+1)} (B(j_1^2/\|\bj\|^2;1/2,-\alpha-1/2)-B(1/2;1/2,-\alpha-1/2))\\
	&\eqspace=  \frac {j_2^{2\alpha+2}}{2^{3/2}(\alpha+1)}{}_2F_1(1/2,3/2+\alpha;3/2;1/2) \\
	&\eqspace\hspace{1em}- \frac {j_2^{2\alpha+2}j_1}{2\|\bj\|(\alpha+1)}{}_2F_1(1/2,3/2+\alpha;3/2;j_1^2/\|\bj\|^2).
\end{align*}
This leads to
\begin{align}
\int_{\tr \bj}\|\bx\|^{2\alpha} \diff \bx
=\ &\frac{j_2^{2\alpha+2}+j_1^{2\alpha+2}}{2^{3/2}(\alpha+1)} {}_2F_1(1/2,3/2+\alpha;3/2;1/2)\nonumber\\
&-\frac{j_1j_2^{2\alpha+2}}{2\|\bj\|(\alpha+1)} {}_2F_1(1/2,3/2+\alpha ;3/2;j_1^2/\|\bj\|^2)\nonumber\\
&-\frac{j_1^{2\alpha+2}j_2}{2\|\bj\|(\alpha+1)}{}_2F_1(1/2,3/2+\alpha ;3/2;j_2^2/\|\bj\|^2),\nonumber
\end{align}
for all $0\leq j_2<j_1$. For implementation we remark that in the case $j_2=0$ the hypergeometric function in the second line is not defined since in this case $j_1^2/\|\bj\|^2=1,$ and we use
\begin{align*}
\int_{\tr (j_1,0)}\|\bx\|^{2\alpha} \diff \bx
&=\frac{\sqrt 2 j_1^{2\alpha+2}}{4(\alpha+1)}{}_2F_1(1/2,3/2+\alpha;3/2;1/2).
\end{align*}

Thus, we have explicit expressions for integrals of the form \eqref{triInt} and all that remains to do is to argue that for $0\leq j_2<j_1$ we can write $C_{\bj,\bj}$ as linear combinations of such integrals.
By symmetry we obtain that
\begin{align*}
 C_{(0,0),(0,0)}=\frac 1 {n^{2+2\alpha}}\int_{\sq}\|\bx\|^{2\alpha}\diff \bx=\frac 8 {n^{2+2\alpha}}\int_{\tr{(1/2,0)'}}\|\bx\|^{2\alpha}\diff \bx.
 \end{align*}
 For  $j>0$ we obtain
 \begin{align*}
C_{(j,j),(j,j)}=&\frac 2 {n^{2+2\alpha}} \int_{\tr{(j+1/2,j-1/2)}}\|\bx\|^{2\alpha}\diff \bx,\quad\text{and}\\
C_{(j,0),(j,0)}=&\frac 2 {n^{2+2\alpha}} \bigg(\int_{\tr{(j+1/2,0)}}\|\bx\|^{2\alpha}\diff \bx-\int_{\tr{(j-1/2,0)'}}\|\bx\|^{2\alpha}\diff \bx\\
&-\int_{\tr{(j+1/2,1/2)}}\|\bx\|^{2\alpha}\diff \bx+\int_{\tr{(j-1/2,1/2)}}\|\bx\|^{2\alpha}\diff \bx\bigg).
\end{align*}
For $0<j_2<j_1$ we obtain
 \begin{align*}
C_{(j_1,j_2),(j_1,j_2)}=&\frac 1 {n^{2+2\alpha}}\bigg(\int_{\tr{(j_1+1/2,j_2-1/2)}}\|\bx\|^{2\alpha}\diff \bx-\int_{\tr{(j_1-1/2,j_2-1/2)}}\|\bx\|^{2\alpha}\diff \bx\\
&-\int_{\tr{(j_1+1/2,j_2+1/2)}}\|\bx\|^{2\alpha}\diff \bx+\int_{\tr{(j_1-1/2,j_2+1/2)}}\|\bx\|^{2\alpha}\diff \bx\bigg).
\end{align*}
This covers all possible choices for $0\leq j_2<j_1$, and consequently we obtain explicit expressions for $C_{\bj,\bj}$ and $C_{\bj,1}$ for all $\bj.$

\bibliographystyle{chicago}

\end{document}